\documentclass[prd,nofootinbib,floats,superscriptaddress,eqsecnum,tightenlines,11pt]{revtex4}

\usepackage{hyperref}
\usepackage{graphicx}
\usepackage{setspace}
\usepackage{amsmath,amssymb,amsfonts,amsthm,latexsym,stmaryrd}
\usepackage{comment}

\def\be{\begin{equation}}
\def\ee{\end{equation}}
\def\bes{\begin{equation*}}
\def\ees{\end{equation*}}
\def\ba{\begin{eqnarray}}
\def\ea{\end{eqnarray}}

\def\Tr{\text{Tr}}

\def\Pexp{\overrightarrow{\exp}}
\def\rd{\mathrm{d}}

\newcommand{\f}{\frac}

\newcommand{\SU}{\text{SU}}

\newcommand{\su}{\mathfrak{su}}
\newcommand{\so}{\mathfrak{so}}

\newtheorem{lemma}{Lemma}
\newtheorem{proposition}{Proposition}

\newtheorem{definition}{Definition}

\begin{document}

\title{Continuous formulation of the Loop Quantum Gravity phase space}

\author{Laurent Freidel}
\email{lfreidel@perimeterinstitute.ca}
\affiliation{Perimeter Institute for Theoretical Physics\\
31 Caroline St. N, N2L 2Y5, Waterloo ON, Canada}

\author{Marc Geiller}
\email{mgeiller@apc.univ-paris7.fr}
\affiliation{Laboratoire APC -- Astroparticule et Cosmologie\\
Universit\'e Paris Diderot Paris 7, 75013 Paris, France}

\author{Jonathan Ziprick}
\email{jziprick@perimeterinstitute.ca}
\affiliation{Perimeter Institute for Theoretical Physics\\
31 Caroline St. N, N2L 2Y5, Waterloo ON, Canada}
\affiliation{Department of Physics, University of Waterloo\\
Waterloo, Ontario N2L 3G1, Canada}

\begin{abstract}
In this paper, we study the discrete classical phase space of loop gravity, which is expressed in terms of the holonomy-flux variables, and show how it is related to the continuous phase space of general relativity. In particular, we prove an isomorphism between the loop gravity discrete phase space and the symplectic reduction of the continuous phase space with respect to a flatness constraint. This gives for the first time a precise relationship between the continuum and holonomy-flux variables. In our construction the fluxes depend not only on the three-geometry, but also explicitly on the connection, providing a natural explanation of their non-commutativity. It also clearly shows that the flux variables do not label a unique geometry, but rather a class of gauge-equivalent geometries. This allows us  to resolve the tension between the loop gravity geometrical interpretation in terms of singular geometry, and the spin foam interpretation in terms of piecewise flat geometry, since we establish that both geometries belong to the same equivalence class. This finally gives us a clear understanding of the relationship between the piecewise flat spin foam geometries and Regge geometries, which are only piecewise-linear flat: While Regge geometry corresponds to metrics whose curvature is concentrated around straight edges, the loop gravity geometry correspond  to metrics whose curvature is concentrated around not necessarily straight edges.
\end{abstract}

\maketitle

\section*{Introduction}

\noindent The classical starting point of Loop Quantum Gravity (LQG) \cite{rovelli-book,thiemann-book} is a Hamiltonian formulation of general relativity in terms of first order connection and triad variables. The basic fields parametrizing the phase space are chosen to be the $\su(2)$-valued Ashtekar-Barbero connection $A$ \cite{barbero}, and its canonically conjugate densitized triad field $E$, both being defined over spatial hypersurfaces foliating the spacetime manifold. The theory comes with a set of first class constraints, namely, the vector constraint generating diffeomorphisms of the spatial hypersurface, the scalar constraint generating time reparametrizations, and the Gauss constraint generating internal $\SU(2)$ gauge transformations.

As a first step towards the construction of the quantum theory, one defines a smearing of the classical Poisson algebra formed by the canonical pair $(A,E)$ by introducing oriented graphs. Given a graph $\Gamma$ embedded in the spatial manifold, the continuous variables $A(x)$ and $E(x)$ are replaced by a pair $(h_e,X_e)$ associated to each edge $e$. The variable $h_e\in\SU(2)$ corresponds to the holonomy of the connection along the edge $e$, and $X_e\in\su(2)$ represents the ``electric'' flux of the densitized triad field across a surface dual to $e$ \footnote{Note that there is an alternative approach where one takes the holonomy-flux variables associated to graphs as fundamental. The notion of a continuous spatial manifold is then seen as an emergent feature of the theory \cite{Zakopane}.}. At the quantum level these new variables form the so-called holonomy-flux algebra \cite{ashtekar-lewandowski}, which is a cornerstone of the entire construction of LQG. The Hilbert space $\mathcal{H}_\Gamma$ of representations associated with this algebra is the so-called spin network Hilbert space. It captures only a finite number of degrees of freedom in the theory. One recovers the continuous kinematical Hilbert space by taking the projective limit of graph Hilbert spaces $\mathcal{H}_\Gamma$. The main challenge is then to formulate a consistent and semi-classically meaningful version of the Hamiltonian constraint acting on the spin network basis.

In this construction, two very different procedures are realized at once. There is a discretization procedure in which the continuous fields are replaced by discrete holonomies and fluxes associated with graphs, and in the same stroke, these variables are promoted into quantum operators. The main idea we want to take advantage of is that the processes of discretization and quantization are totally independent.
In this work we would like to disentangle these two steps. We propose to study only the process of discretization using graphs, without delving into the quantization of the theory. This means that we first associate to a given graph a finite-dimensional holonomy-flux phase space generated by $(h_e,X_e)\in T^*\SU(2)$. The phase space of loop gravity on a graph is obtained as a direct product over the edges of $\SU(2)$ cotangent bundles. The main point of the present paper is to understand the exact relationship between this finite-dimensional discrete phase space, and the continuous phase space of gravity. We show explicitly that an element of the discrete phase space represents a specific equivalence class of continuous geometries.

The advantage of considering classical loop gravity is threefold. First, it provides a truncation of the classical phase space of gravity in terms of finite-dimensional holonomy-flux phase spaces, whose quantizations are given by spin network states. Second, it allows us to shed some light on the geometrical interpretation of the holonomy-flux variables, and the type of geometry that they represent. For instance, we will see that both the singular geometry of LQG and the piecewise flat geometry of spin foam models are represented by the same flux data as two representatives of the same equivalence class. As we will see in the end, our result also allows us to understand more precisely the relationship between the spin foam geometrical interpretation and Regge geometry. Namely, it shows that twisted geometries \cite{freidel-speziale} described by fluxes can be understood as piecewise flat geometries which are not necessarily piecewise-linear flat, as is the case for Regge geometry \cite{regge,dittrich-ryan}. Finally, this approach is designed to allow us to address at the classical level one of the most challenging questions of LQG: Is it possible to express, in the classical setting, the dynamics of general relativity in terms of a collection of truncated dynamics between finite-dimensional phases spaces parametrized by holonomies and fluxes. In other words, can we capture the full dynamics of gravity in terms of the holonomy-flux phase spaces if we simultaneously consider \emph{all} graphs. If there is a clear positive answer to this question at the classical level, then the quantization of loop gravity will be reduced to the treatment of quantization ambiguities in a finite-dimensional system. If, on the other hand, we get a negative answer at the classical level, then no quantization in terms of holonomy-flux variables can express the quantum gravitational dynamics. It is therefore of utmost importance to eventually understand the classical dynamics of general relativity in terms of the holonomy-flux representation.

Let us stress that the classical picture of the loop gravity phase space that we develop here is, when quantized, related to the picture first proposed by Bianchi in \cite{bianchi}. In this precursor work, it is argued that the spin network Hilbert space can be identified with the state space of a topological theory on a flat manifold with defects. Our analysis makes the same type of identification at the classical level and emphasizes the fact that the frame field determines only an equivalence class of geometries. The idea that the discrete data labels only an equivalence class of geometries has already been advocated in \cite{rovelli-speziale} on a general basis. Our approach gives a precise understanding of which set or equivalence class of continuous geometries is represented by the discrete geometrical data.

We begin in section \ref{sec1} by defining the continuous phase space of gravity in terms of the connection and triad variables $A$ and $E$, and recall some facts about the process of symplectic reduction. In section \ref{sec2} we introduce the discrete classical spin network phase space associated to a graph. In particular, we explain how to obtain the discrete data $(h_e,X_e)$ starting from the continuous fields $A$ and $E$, and
showing the fluxes cannot depend only on $E$ but need to involve the connection in their definition. This construction explains why the flux variable carries information about both intrinsic and extrinsic geometry, in agreement with what has been pointed out already in \cite{freidel-speziale}. In section \ref{sec3}, we prove that the discrete holonomy-flux phase space can be obtained as a symplectic reduction of the continuous phase space. This shows that the discrete data corresponds to an equivalence class of continuous three-geometries related by gauge transformations. In section \ref{sec4} we show that given a particular gauge choice, the discrete data can be used to reconstruct a configuration of the continuous fields. We will show in particular that it is possible to represent a given equivalence class of geometries by either a singular gauge choice in agreement with the LQG interpretation of polymer geometry, or a flat gauge choice corresponding to the geometrical interpretation of spin foams. Finally, in section \ref{sec5} we discuss the notion of cylindrical consistency and cylindrical operators, and explain how it is possible to relate operators constructed on the discrete and the continuous phase spaces.

Notations are such that $\mu,\nu,\dots$ refer to spacetime indices, $a,b,\dots$ to spatial indices, $I,J,\dots$ to Lorentzian indices, and $i,j,\dots$ to $\su(2)$ indices. We will assume that the four-dimensional spacetime manifold is topologically $\Sigma\times\mathbb{R}$, where $\Sigma$ is a three-dimensional manifold without boundaries.


\section{Continuous phase space of gravity}
\label{sec1}

\noindent The loop approach to quantum gravity relies on the well-known idea that the phase space of Lorentzian or Riemannian general relativity can be parametrized in terms of an $\su(2)$-valued connection one-form $A^i_a$ and a densitized triad field $\tilde{E}^a_i$, both fields being defined over a base three-dimensional spacetime manifold $\Sigma$ (which here we assume to be isomorphic to $\mathbb{S}^3$). The $\su(2)$ Ashtekar-Barbero connection $A^i_a$ is related to the spacetime $\so(3,1)$ spin connection $\omega_\mu^{IJ}$ and to the geometrodynamical variables of the ADM phase space via
\be
A^i_a\equiv \f{1}{2}\epsilon^i_{~jk}\omega^{jk}_a+\gamma\omega^{0i}_a=\Gamma^i_a+\gamma K^i_a,
\ee
where $\gamma\in\mathbb{R}-\{0\}$ is the Barbero-Immirzi parameter, $\Gamma^i_a$ is the Levi-Civita spin connection, and $K^i_a$ the extrinsic curvature one-form. The densitized triad and the three-dimensional frame field $e^i_a$ are related by
\be
\tilde{E}_i^a=\f{1}{2}\epsilon^{abc}\epsilon_{ijk}e^j_be^k_c,\qquad\qquad\det(e)e^i_a=\f{1}{2}\epsilon^{ijk}\epsilon_{abc}\tilde{E}^b_j\tilde{E}^c_k.
\ee
These variables form the Poisson algebra
\be\label{poisson algebra}
\big\lbrace A^i_a(x),A^j_b(y)\big\rbrace=\big\lbrace\tilde{E}^a_i(x),\tilde{E}^b_j(y)\big\rbrace=0,\qquad\qquad\big\lbrace A^i_a(x),\tilde{E}^b_j(y)\big\rbrace=\gamma \delta^i_j\delta^b_a\delta^3(x-y).
\ee

The classical configuration space of the theory is the space $\mathcal{A}$ of smooth connections on $\Sigma$. The phase space is the cotangent bundle $\mathcal{P}\equiv T^*\mathcal{A}$, and carries a natural symplectic potential. In the following we will denote by $E_{ab\, i}$ (without tilde) the Lie algebra-valued two-form related to the densitized vector $\tilde{E}_a^i$ through
\be
E_{ab\,i}\equiv\epsilon_{abc}\tilde{E}^c_i,\qquad\qquad E_i\equiv E_{ab\,i}\rd x^a\wedge\rd x^b.
\ee
The symplectic potential of the cotangent bundle  is given by
\be\label{symplectic gravity}
\Theta=\int_\Sigma E_i\wedge\delta A^i=\int_\Sigma\Tr\left(E\wedge\delta A\right),
\ee
where we denote by $\Tr$ the natural metric on $\su(2)$ which is invariant under the adjoint action $\text{Ad}_{\SU(2)}$ of the group. The phase space $\mathcal{P}$ also carries an action of the gauge group $\SU(2)$ and of spatial diffeomorphisms. In fact, since $\mathcal{P}$ is $18$-dimensional at each point of $\Sigma$, the (first class) constraints of the canonical theory have to be taken into account in order to obtain the physical phase space with $4$ degrees of freedom at each point. This can be achieved through the process of symplectic (or Hamiltonian) reduction, which we now describe.

Let $P$ be a symplectic manifold, which is seen as the classical phase space of the theory of interest, and $G$ a group of transformations. Suppose that the infinitesimal group transformations are generated via Poisson bracket by a Hamiltonian $H$. Then the Marsden-Weinstein theorem \cite{marsden-weinstein,butterfield,weinstein} ensures that the symplectic reduction of $P$ by the group $G$, denoted by the double quotient $P\sslash G$, is still a symplectic manifold and carries a unique symplectic form. The reduced phase space is given by imposing the constraints and dividing the constraint surface by the action of gauge transformations. This is written as
\be
P\sslash G\equiv H^{-1}(0)/G.
\ee
For notational simplicity, we will denote the group of transformations $G$ and the associated Hamiltonian $H$ with the same letters. Note that the Marsden-Weinstein theorem is proven in general for finite-dimensional phase spaces, but these methods are commonly extended to infinite-dimensional phase spaces. See \cite{AandB} for a symplectic reduction of $\mathcal{P}\equiv T^*\mathcal{A}$, and the first two chapters of \cite{Dirac} for a description of this method as commonly employed in physics.

In the case of four-dimensional gravity, the physical phase space is obtained from the kinematical (unconstrained) phase space $\mathcal{P}$ by performing three symplectic reductions. The first one is defined with respect to the group of $\SU(2)$ gauge transformations $\mathcal{G}\equiv C^\infty\big(\Sigma,\SU(2)\big)$. Since the action of this gauge group on $\mathcal{P}$ is Hamiltonian, we can define the gauge-invariant phase space $T^*\mathcal{A}\sslash\mathcal{G}$. More precisely, the Hamiltonian generating these transformations is the smeared Gauss constraint:
\be
\mathcal{G}(\alpha)=\int_\Sigma \alpha^i (\rd_A E)_i=0,
\ee
where $\rd_{A}$ denotes the covariant differential and $\alpha$ is a Lie algebra-valued function. Its infinitesimal action on the phase space variables is given by
\be
\delta^\mathcal{G}_\alpha A=\big\lbrace A,\mathcal{G}(\alpha)\big\rbrace=\rd_A \alpha,\qquad\qquad\delta^\mathcal{G}_\alpha E=\big\lbrace E,\mathcal{G}(\alpha)\big\rbrace=[E,\alpha].
\ee
The other relevant symplectic reduction is defined with respect to the group of spatial diffeomorphisms, and enables one to construct the diffeomorphism-invariant phase space $T^*\mathcal{A}\sslash\big(\mathcal{G}\times\text{Diff}(\Sigma)\big)$. Here, the action of the group of diffeomorphisms on the phase space variables is given by
\be
\delta^\mathcal{D}_\xi A=\big\lbrace A,\mathcal{D}(\xi)\big\rbrace=\mathcal{L}_\xi A,\qquad\qquad\delta^\mathcal{D}_\xi E=\big\lbrace E,\mathcal{D}(\xi)\big\rbrace=\mathcal{L}_\xi E,
\ee
where $\mathcal{L}_\xi$ is the Lie derivative along the vector field $\xi^a$. This group is generated through Poisson brackets with the Hamiltonian
\be
\mathcal{D}(\xi)=\mathcal{H}(\xi)-\mathcal{G}(\xi^{a}A_{a}),\qquad\qquad\text{with}\qquad\mathcal{H}(\xi)=\int_\Sigma\xi^aF^i_{ab}E^b_i.
\ee
 Finally, the physical phase space can be obtained from the gauge and diffeomorphism-invariant phase space by performing a symplectic reduction with respect to the scalar constraint. This latter is given by
\be
\mathcal{H}(N)=\int_\Sigma N\f{E^a_iE^b_j}{2\sqrt{\det(E^a_i)}}\left(\epsilon^{ij}_{~~k}F^k_{ab}-2(\gamma^2-\sigma)K^i_{[a}K^j_{b]}\right)=0,
\ee
where the smearing variable is the lapse function $N$, and $\sigma=\mp1$ in Lorentzian or Riemannian signature respectively. Notice that for a (anti) self-dual connection ($\gamma=\pm i$ in the Lorentzian case, or $\pm1$ in the Riemannian case) the second term vanishes and the constraint simplifies greatly.

\section{Spin network phase space}
\label{sec2}

\noindent In loop gravity, one does not work directly with the continuous kinematical Hilbert space, but instead with the projective limit of Hilbert spaces associated to embedded oriented graphs $\Gamma$ \cite{ashtekar-lewandowski,AL}. The Hilbert space associated with one graph is the so-called spin network Hilbert space. It represents a truncation of the full Hilbert space to a finite number of degrees of freedom. What we would like to emphasize here  is that spin network Hilbert spaces can be obtained as the quantization of finite-dimensional phase spaces associated to embedded oriented graphs $\Gamma$. Each of these truncated phase spaces are spanned by a finite number of holonomies and fluxes which reproduce the Poisson algebra of $T^*\SU(2)$. This fact has already been recognized in the literature \cite{rovelli-speziale} and is at the basis of most of the recent semi-classical analyses of LQG \cite{coherent,coherent2,coherent3}. Our main point is that the process of truncating the theory to a finite number of degrees of freedom and the process of quantizing this truncated theory are separate constructions which have to be studied individually. Here we would like to adopt the point of view that the continuous  kinematical phase space $\mathcal{P}$ can be described as the projective limit of phase spaces $P_{\Gamma}$ associated to embedded oriented graphs $\Gamma$. In particular, we would like to understand the relationship between these finite-dimensional phase spaces $P_{\Gamma}$ and the continuous phase space variables.

In \cite{thiemann-cont1} it has been shown how, for a given graph $\Gamma$, the graph phase space $P_\Gamma$ can be obtained from the continuous phase space, and carries the Poisson structure of finite direct products of $\SU(2)$ cotangent bundles. It has furthermore been shown how the regulator corresponding to the graph can be removed, thereby defining a continuum limit (via a suitable projective sequence) which leads back to the original infinite-dimensional continuous phase space $T^*\mathcal{A}$. While in the present work we will recall some elements of this construction like the definition of the discrete spin network phase spaces, our new message is to show how, without taking the continuum limit, it is possible to understand the discrete holonomy and flux elements as labels for particular configurations on the  original phase space parametrized by continuous fields $A(x)$ and $E(x)$.

An oriented graph $\Gamma$ is defined as a one-cellular complex \cite{rourke-sanderson} consisting of a set $E_\Gamma$ of oriented edges $e$ (one-dimensional analytic submanifolds of $\Sigma$) and a set $V_\Gamma$ of vertices $v$. The end points of the oriented edges are the vertices, and we denote by $s,t$ the two functions assigning a source vertex $s(e)$ and a target vertex $t(e)$ to each edge $e$.
We also denote by $e^{-1}$ the edge $e$ with a reverse orientation.
 The kinematical spin network phase space $P_\Gamma$ associated with such a graph is isomorphic to a direct product for each edge of $\SU(2)$ cotangent bundles\footnote{Given a Lie group $G$, the group action on itself by left (or right) multiplication can be used to obtain an isomorphism of vector fields with the Lie algebra $\mathfrak{g}$, and to trivialize the cotangent bundle as $T^*G=G\times\mathfrak{g}^*$ \cite{fuchs-schweigert}.}:
\be
P_\Gamma\equiv\underset{e}{\times}T^*\SU(2)_e.
\ee
Explicitly, this phase space is labeled by couples $(h_e,X_e)\in\SU(2)\times\su(2)$ of Lie group and Lie algebra elements for each edge $e\in\Gamma$. This data depends on a choice of orientation for each edge, and under an orientation reversal we have
\be\label{orientation}
h_{e^{-1}}=h_e^{-1},\qquad\qquad X_{e^{-1}}=-h_e^{-1}X_e h_e.
\ee
Since we have chosen here to trivialize $T^*\SU(2)$ with right-invariant vector fields, this last relation means that under orientation reversal of the edge we obtain the left-invariant ones. The variables $(h_e,X_e)$ satisfy the Poisson algebra
\be\label{poisson algebra2}
\big\lbrace X^i_e,X^j_{e'}\big\rbrace=\delta_{ee'}\epsilon^{ij}_{~~k}X^k_e,\qquad\big\lbrace X^i_e,h_{e'}\big\rbrace=-\delta_{ee'}\tau^ih_e+\delta_{ee'^{-1}}\tau^ih_{e}^{-1},\qquad\big\lbrace h_e,h_{e'}\big\rbrace=0,
\ee
where we have used notations such that\footnote{In this work, we define $\tau_i=-i\sigma_i/2$, where $\sigma_i$ are the Pauli matrices. The $\su(2)$ commutation relations are then given by $[\tau_i,\tau_j]=\epsilon_{ij}^{~~k}\tau_k$, where $\epsilon_{ij}^{~~k}$ is the completely antisymmetric Levi-Civita tensor.} $X_e \equiv X^i_e\tau_i$. As shown in \cite{arnold,alekseev}, the symplectic potential and symplectic two-form for this Poisson structure are given respectively by
\be\label{symplectic}
\Theta_\Gamma=\sum_e\Tr\left(X_e \delta h_eh_e^{-1}\right),\qquad\qquad\Omega_\Gamma=-\rd\Theta_\Gamma.
\ee

On the spin network phase space $P_\Gamma$, we can define the action of the gauge group $G_\Gamma\equiv\SU(2)^{|V_\Gamma|}$ at the vertices $V_\Gamma$ of the graph. Given an element $g_v\in\SU(2)$, finite gauge transformations are given by
\be\label{SU(2)}
g_v\triangleright h_e=g_{s(e)}h_eg_{t(e)}^{-1},\qquad\qquad g_v\triangleright X_e=g_{s(e)}X_eg_{s(e)}^{-1},
\ee
where $s(e)$ (resp. $t(e)$) denotes the starting (resp. terminal) vertex of $e$. This action on the variables $h_e$ and $X_e$ is generated at each vertex by the Hamiltonian
\be\label{discrete gauss}
G_v\equiv\sum_{e\ni v}X_e=\sum_{e|s(e)=v}X_{e}+\sum_{e|t(e)=v}X_{e^{-1}},
\ee
which can be understood as a discrete Gauss constraint. Since this action is Hamiltonian, we can define the gauge-invariant phase space
\be
P^G_\Gamma=\underset{e}{\times}T^*\SU(2)_e\sslash\SU(2)^{|V_\Gamma|}=G_v^{-1}(0)/\SU(2)^{|V_\Gamma|}
\ee
by symplectic reduction where, as explained above, the double quotient means imposing the Gauss constraint at each vertex $v$ and then dividing out the action of the $\SU(2)$ gauge transformation (\ref{SU(2)}) that it generates.

The question we would like to address is: What is the relationship between the continuous phase space $\mathcal{P}$ described in the previous section, and the spin network phase space $P_\Gamma$? More precisely, we would like to know if it is possible to reconstruct from the discrete data $P_\Gamma$ a point in the continuous phase space $\mathcal{P}$? In order to describe the relationship between the discrete and continuous data, we need a map from the continuous to the discrete phase space. We can then study its kernel and see to what extent it can be inverted. This is the object of the next sections.

\subsection{From continuous to discrete data}

\noindent In order to construct the discrete data, let us first choose an embedding $f_\Gamma:\Gamma\longrightarrow\Sigma$ of the graph $\Gamma$ into the spatial manifold $\Sigma$. Given this embedding, it is well understood in the discrete picture that the group elements $h_e$ represent holonomies of the Ashtekar-Barbero connection $A^i_a$ along edges $e$. It is necessary to work with such objects because an important step toward the quantization of the canonical theory is the smearing of the Poisson algebra (\ref{poisson algebra}). Since the connection $A^i_a$ is a one-form, it is natural to smear it along paths $e$.
Now we could just take the integral of $A$ along $e$ as a smearing but this will not respect the gauge transformations. What is needed is a smearing that does intertwine the notion of continuous and discrete gauge transformations. It is well known that this is given by
the notion of parallel transport along $e$, encoded in the holonomy
\be
h_e(A)\equiv\Pexp\int_eA=\Pexp\int_eA^i_a\dot{e}^a\tau_i=\Pexp\int_{s(e)}^{t(e)}A^i_a\dot{e}^a\tau_i,
\ee
where $\dot{e}^a$ denotes the tangent vector to the path and $\Pexp$ denotes the path-ordered exponential.

Let us recall some fundamental properties of the holonomy functional. The holonomy is invariant under reparametrizations of the path $e$, and the holonomy of a path corresponding to a single point is the identity. If we consider the composition $e=e_1\circ e_2$ of two paths which are such that $s(e_2)=t(e_1)$, the holonomy satisfies
\be
h_e=h_{e_1}h_{e_2}.
\ee
If we reverse the orientation of a path, we have
\be
h_{e^{-1}}=h_e^{-1}.
\ee
These properties come from the fact that the holonomy is a representation of the groupoid of oriented paths \cite{baez-huerta}. Under $\SU(2)$ gauge transformations, the holonomy transforms as
\be
g\triangleright h_e=g_{s(e)}h_eg_{t(e)}^{-1},
\ee
which shows that the finite gauge transformation $g\,\triangleright A=gAg^{-1}+g\rd g^{-1}$ of the connection becomes a discrete gauge symmetry acting on the vertices defining the boundary of the edge $e$. Finally, under the action of a diffeomorphism $\Phi\in\text{Diff}(\Sigma)$, the holonomy transforms as
\be
h_e(\Phi^*A)=h_{\Phi(e)}(A).
\ee

The exact meaning of ``momentum'' variable $X_e$ is less clear. Roughly speaking, we usually build a flux operator by smearing the field $E^a_i$ along a surface $F_e$ dual to an edge $e$ \cite{thiemann-book}. But if one wants this integrated flux to have a covariant behavior under gauge transformations, it is essential for the integration along $F_e$ to involve some notion of parallel transport. Indeed, the naive definition
\be\label{usual flux}
\bar{X}_e(E)=\int_{F_e}E(x)
\ee
of the flux is not covariant under gauge transformations, i.e.
\be
\bar{X}_e(g\triangleright E)=\bar{X}_e\left(gEg^{-1}\right)\neq g_{s(e)}\bar{X}_e(E)g_{s(e)}^{-1}.
\ee
This is an important point which has often been ignored in the LQG literature, the only noticeable exceptions being \cite{thiemann-cont1,sahlmann-thiemann}, and more recently \cite{freidel-speziale,wieland}.
For the holonomy, the only reason we consider the parallel transport operator instead of the simple integral of $A$ along $e$ is to have a discretization covariant under gauge transformation.
It is as important to preserve this covariance for the flux as it is for the holonomy. Another drawback is that the non-covariant definition of the flux does not produce the Poisson algebra given in (\ref{poisson algebra2}), unless the intersection $F_e\cap e$ between face and edge is at the start point $s(e)$ or terminal point $t(e)$ of the edge. If we consider a face that intersects somewhere in the middle of the edge, i.e. write the edge as $e=e_1\circ e_2$ and have the intersection $F_e\cap e=s(e_2)=t(e_1)$, then we have
\be
\big\lbrace \bar{X}^i_e,h_{e'}\big\rbrace=-\delta_{ee'}h_{e_1}\tau^i h_{e_2}+\delta_{ee'^{-1}}h_{e_1}^{-1}\tau^i h_{e_2}^{-1},
\ee
which splits the holonomy in two.

The way around this problem is to define a flux operator which also depends on the connection through its holonomy. Given an oriented edge $e\in\Gamma$ and a point $u$ on this edge, we choose a surface $F_e$ intersecting $e$ transversally at $u=F_e\cap e$. We also choose a set of paths $\pi_e$ assigning to any point $x\in F_e$ a unique path $\pi_e$ going from the source $s(e)$ to $x$. Such a path starts at the source vertex of the edge $e$, goes along $e$ until it reaches the intersection point $u=F_e\cap e$, and then goes from $u$ to any point $x\in F_e$ while staying tangential to the surface $F_e$. More precisely, we have $\pi_e:F_e\times[0,1]\longrightarrow\Sigma$ such that $\pi_e(x,0)=s(e)$ and $\pi_e(x,1)=x$. With the set of data $(F_e,\pi_e)$, one can define the flux operator
\be\label{new flux}
X_{(F_e,\pi_e)}(A,E)\equiv\int_{F_e}h_{\pi_e}(x)E(x)h_{\pi_e}(x)^{-1},
\ee
where
\be
h_{\pi_e}(x)\equiv\Pexp\int_{s(e)}^xA.
\ee
Notice that by definition, the source of the path $\pi_e$ is $s(e)$, and its target is the point $x\in F_e$. Therefore, under the gauge transformations
\be
g\triangleright E(x)=g(x)E(x)g(x)^{-1},\qquad\qquad g\triangleright h_{\pi_e}(x)=g_{s(e)}h_{\pi_e}(x)g(x)^{-1},
\ee
the flux operator becomes
\be
X_{(F_e,\pi_e)}(g\triangleright A,g\triangleright E)=g_{s(e)}X_{(F_e,\pi_e)}(A,E)g_{s(e)}^{-1},
\ee
which is in agreement with (\ref{SU(2)}). The existence of a covariant transformation property is one of the main justifications for introducing the extra holonomy dependance in the definition of the flux operator. With the definition (\ref{new flux}), the flux operator intertwines the continuous and discrete actions of the gauge group.

From the definition of the paths $\pi_e$ we see that reversing the orientation of the edge gives a system of paths beginning at $t(e)$ and ending at a point $x\in F_e$, i.e. $\pi_{e^{-1}}(x,0)=t(e)$ and $\pi_{e^{-1}}(x,1)=x$. This implies that $\pi_{e^{-1}}=e^{-1}\circ\pi_e$, and therefore
\be
h_{\pi_{e^{-1}}}(x)=h_e^{-1}h_{\pi_e}(x).
\ee
Moreover, the surface $F_e$ possesses a reverse orientation $F_{e^{-1}}=-F_e$, and thus we have
\be
X_{\left(F_{e^{-1}},\pi_{e^{-1}}\right)}=-h_e^{-1}X_{(F_e,\pi_e)}h_e,
\ee
which proves that our mapping is consistent with (\ref{orientation}). Notice also that any two fluxes that differ only by the choice of surfaces are in the commutant of the holonomy algebra:
\be
\big\lbrace X_{(F'_e,\pi'_e)}-X_{(F_e,\pi_e)},h_e\big\rbrace=0,
\ee
where $\pi_e$ and $\pi'_e$ each follow the edge until the intersection points with their respective surfaces as defined above. An important feature of the mapping that we have described is that it reproduces the Poisson algebra (\ref{poisson algebra2}), specifically the Poisson bracket between flux and holonomy. To show this, let us use the notation $R(h_{\pi_e})^i_j E^j \equiv (h_{\pi_e} E h_{\pi_e}^{-1})^i$ in writing the flux. In the case where the flux and holonomy are associated to the same edge and have the same orientation, we have
\ba
\big\lbrace X^i_{(F_e,\pi_e)}(A,E),h_{e}(A)\big\rbrace&=&\int_{F_e}R(h_{\pi_e})^i_j\left\{E^j,h_{e}(A)\right\}\nonumber\\
&=&-R(h_{e_1})^i_j h_{e_1}\tau^jh_{e_2}\nonumber\\
&=&-h_{e_1}h_{e_1}^{-1}\tau^ih_{e_1} h_{e_2}\nonumber\\
&=&-\tau^ih_e,
\ea
where we are using the same notation as above when splitting the edge into $e_1$ and $e_2$ at the point of intersection, and we have $h_{\pi_e} = h_{e_1}$ at the only point contributing to the integral in the first line. Also notice that the $\SU(2)$ rotation $R(h_{\pi_e})^i_j$ acts on the basis elements $\tau^i$ inversely to the way it acts on $E^i$. A similar calculation with the inverse holonomy yields the second term shown in (\ref{poisson algebra2}).

Finally, we know that the requirement of consistency with the Jacobi identity imposes that the fluxes do not commute among each other. This property, which seems inconsistent if $X_{e}$ depends purely on the (commuting) densitized triad field, is perfectly understandable if the flux depends also on the connection, and provides a natural explanation to the ``mystery'' behind the non-commutativity of the fluxes \cite{ashtekar-zapata}.  This is consistent with the understanding of the spin network phase space in terms of twisted geometries \cite{freidel-speziale}, where it appears clearly that the flux operators also contain information about the holonomies, and cannot be thought of as being purely geometrical. In other words, the flux operators are not commuting because they capture information not only about the intrinsic geometry, but also about the extrinsic curvature.

The map that we have described depends on three types of data. It depends on a choice of embedding $f_\Gamma$ of $\Gamma$ into $\Sigma$, a choice of surface $F_e$ transverse to the edge $e$ at $u$, and a choice of path $\pi_e$ going from $s(e)$ to a point $x\in F_e$. Once this data is given, we can construct a map
\be
\begin{array}{cccl}
I:&\mathcal{P}&\longrightarrow&P_\Gamma\\
&(A,E)&\longmapsto&\left(h_e(A),X_{(F_e,\pi_e)}(A,E)\right),
\end{array}
\ee
which has the key property of intertwining gauge transformations on the continuous and discrete phase spaces, is compatible with the orientation reversal of the edges, and respects the Poisson structure of $T^*\SU(2)$.

\subsection{From discrete to continuous data}

\noindent Now we would like to investigate to what extent it is possible to invert the map from continuous to discrete data $I:\mathcal{P}\longrightarrow P_\Gamma$. In other words, to what extent does the discrete data determine the continuous data? Can we reconstruct a unique representative of the continuous data starting from the discrete one, or describe a specific equivalence class?

At first sight, this seems like an impossible task. Indeed, if one first focuses on the connection, one needs to choose an embedding $f_\Gamma$ to construct the holonomies, so there is no way the discrete group elements will determine the connection unless we know this embedding. Moreover, one clearly sees that the flux operator is not uniquely defined by the electric field $E$. There are several ambiguities in its definition. There are many possible choices of surfaces $F_e$ that are transverse to the edge $e$, and also many possible paths that one can choose on $F_e$. Different choices lead to different mappings from the continuous data to the discrete data. This means that giving a flux $X_{(F_e,\pi_e)}$ (which we will call $X_e$ for simplicity) does not allow one to reconstruct a continuous field $E$, which constitutes a fundamental ambiguity. This state of affairs is fine if one treats the discrete data as some approximate description of continuous geometry which only takes physical meaning in some continuous limit. This is the usual point of view \cite{rovelli-speziale}, and it implies that operators expressed in terms of the fluxes $X_e$ do not have a sharp semi-classical geometric interpretation.

In this work we would like to be more ambitious and interpret the discrete data as potential initial value data for the continuous theory of gravity. The challenge is to show that one can reconstruct continuous fields $(A,E)$ explicitly from the knowledge of the discrete data $(h_e,X_e)$. How can this be possible in light of all the ambiguities that we have listed above? In order to make some progress in this direction, let us first remark that there are configurations of fields for which the ambiguities disappear. This is the case in particular for a flat connection.

Suppose that we focus on a region $C_v$ of simple topology (isomorphic to a three-ball) around a vertex $v\in C_v$, and that in this region the connection $A$ is flat. In this case, the expression (\ref{new flux}) for the flux becomes independent of the system of paths $\pi_e$, since the flatness of the connection implies that there exists an $\SU(2)$ element $a(x)$ such that $A=a\rd a^{-1}$ and $h_{\pi_e}(x)=a(v)a(x)^{-1}$. Indeed, we have
\be
X_{(F_e,\pi_e)}=X_e=a(v)\left(\int_{F_e}a(x)^{-1}E(x)a(x)\right)a(v)^{-1},
\ee
and the dependence on the system of paths has disappeared. Moreover, one can see that the Gauss law expresses the fact that $X^i_{F_e}=X^i_{F'_e}$, for if $F_e$ and $F'_e$ have the same oriented boundary, their union encloses a volume $\partial C_v$ and we have that:
\be\label{same flux}
0=\int_{C_v}a(x)^{-1}\rd_AE(x)a(x)=\int_{C_v}\rd\left(a(x)^{-1}E(x)a(x)\right)=a(v)^{-1}\left(X_{F_e}-X_{F'_e}\right)a(v).
\ee
In the next section, we are going to make this statement more precise, and study the case of a partially flat connection.

\section{Partially flat connection}
\label{sec3}

\noindent In this section, we formulate and prove the equivalence between the continuous phase space of partially flat geometries and the discrete spin network phase space. In order to do so, we first need to introduce some notions of topology.

\begin{definition}
A cellular decomposition $\Delta$ of a space $\Sigma$ is a decomposition of $\Sigma$ as a disjoint union (partition) of open cells of varying dimension satisfying the following conditions:\par
i) An $n$-dimensional open cell is a topological space which is homeomorphic to the $n$-dimensional open ball.\par
ii) The boundary of the closure of an $n$-dimensional cell is contained in a finite union of cells of lower dimension.\par
The $n$-skeleton $\Delta_n$ of a cellular decomposition is the union of cells of dimension less than or equal to $n$.
\end{definition}

Clearly, the $n$-skeleton of a cellular decomposition is also a cellular decomposition. In particular, the one-skeleton $\Delta_1$ of a cellular decomposition is a graph. Let us now suppose that we have a graph $\Gamma$ embedded in $\Sigma$. We need to introduce the notion of a cellular decomposition dual to $\Gamma$.

\begin{definition}
A cellular decomposition $\Delta$ of a three-dimensional space $\Sigma$ is said to be dual to the graph $\Gamma$ if there is a one-to-one correspondence $v\longmapsto C_v$ between vertices of $\Gamma$ and three-cells of $\Delta$, and a one-to-one correspondence $e\longmapsto F_e$ between edges of $\Gamma$ and two-cells of $\Delta$, such that:\par
i) There is a unique vertex $v$ inside each three-cell $C_v$.\par
ii) The two-cells $F_e$ intersect $\Gamma$ transversally at one point only, and the intersection belongs to the interior of the edge $e$ of $\Gamma$.
\end{definition}

In other words, a cellular decomposition dual to $\Gamma$ is such that each vertex of $\Gamma$ is dual to a three-cell, and each edge of $\Gamma$ is dual to a two-cell. Finally, let us consider a pair $(\Gamma,\Gamma^*)$ of graphs embedded in $\Sigma$.

\begin{definition}
We say that an embedded graph $\Gamma^*$ is dual to the embedded graph $\Gamma$ (and vice versa), or that $(\Gamma,\Gamma^*)$ forms a pair of dual graphs, if there exists a cellular decomposition $\Delta$ dual to $\Gamma$, whose one-skeleton $\Delta_1$ is $\Gamma^*$.
\end{definition}

From now on, we consider that $(\Gamma,\Gamma^*)$ is a pair of dual embedded graphs, and we denote by $\Delta$ the cellular decomposition dual to $\Gamma$ with a one-skeleton $\Delta_1$ given by $\Gamma^*$. Notice that if we take any diffeomorphism $\Phi_{\!o}$ on $\Sigma$ which does not act on $\Gamma^*$ or the vertices of $\Gamma$, we obtain an equivalent\footnote{Since these diffeomorphisms vanish on $\Gamma^*$, the duality between edges and faces is preserved.} cellular decomposition $\Phi_{\!o}(\Delta)$. Given such a pair of dual graphs, we are going to construct a certain phase space $\mathcal{P}_{\Gamma, \Gamma^*}$, and prove that it is the continuous analogue of the discrete spin network phase space $P_\Gamma$. In fact, we are going to show that there is a symplectomorphism between $\mathcal{P}_{\Gamma, \Gamma^*}$ and $P_\Gamma$.

\subsection{The reduced phase space $\mathcal{P}_{\Gamma, \Gamma^*}$}
\label{subsection flat}

\noindent To define the reduced phase space $\mathcal{P}_{\Gamma, \Gamma^*}$, we first construct a group $\mathcal{F}_{\Gamma^*}\times\mathcal{G}_\Gamma$ of gauge transformations acting on $\mathcal{P}$. For this, let us consider an infinite-dimensional Abelian group of transformations $\mathcal{F}_{\Gamma^*}$ parametrized by Lie algebra-valued one-forms $\phi^i\in\Omega^1\big(\Sigma,\su(2)\big)$ which have the property that they vanish on $\Gamma^*$:
\be
\phi^i(x)=0,\qquad\qquad\forall\:x\in\Gamma^*.
\ee
This group action is Hamiltonian and generated by the curvature constraint
\be\label{curvature constraint}
\mathcal{F}_{\Gamma^*}(\phi)=\int_\Sigma\phi_i\wedge F^i(A),
\ee
whose action on the continuous phase space $\mathcal{P}$ is given by
\be\label{flat-gauge}
\delta^{\mathcal{F}_{\Gamma^*}}_\phi A=\big\lbrace A,\mathcal{F}_{\Gamma^*}(\phi)\big\rbrace=0,\qquad\qquad\delta^{\mathcal{F}_{\Gamma^*}}_\phi E=\big\lbrace E,\mathcal{F}_{\Gamma^*}(\phi)\big\rbrace=\rd_A\phi.
\ee
This constraint enforces the flatness of the connection outside of the one-skeleton graph $\Gamma^*$. See
The second group, $\mathcal{G}_\Gamma$, is the group of gauge transformations parametrized by Lie algebra-valued functions $\alpha^i\in\Omega^0\big(\Sigma,\su(2)\big)$ which have the property that they vanish on the vertices of $\Gamma$:
\be
\alpha^i(x)=0,\qquad\qquad\forall\:x\in V_\Gamma.
\ee
This group action is also Hamiltonian. It is generated by the smeared Gauss constraint
\be\label{gauss hamiltonian}
\mathcal{G}_\Gamma(\alpha)=\int_\Sigma\alpha^i(\rd_A E)_i,
\ee
whose infinitesimal action on the phase space variables is given by
\be\label{SU(2)-gauge}
\delta^{\mathcal{G}_\Gamma}_\alpha A=\big\lbrace A,\mathcal{G}_\Gamma(\alpha)\big\rbrace=\rd_A\alpha,\qquad\qquad\delta^{\mathcal{G}_\Gamma}_\alpha E=\big\lbrace E,\mathcal{G}_\Gamma(\alpha)\big\rbrace=[E,\alpha].
\ee
From the various Poisson brackets
\begin{subequations}
\ba
\big\lbrace\mathcal{G}_\Gamma(\alpha),\mathcal{G}_\Gamma(\alpha')\rbrace &=& \mathcal{G}_\Gamma([\alpha,\alpha']),\\
\big\lbrace\mathcal{G}_\Gamma(\alpha),\mathcal{F}_{\Gamma^*}(\phi)\rbrace &=& \mathcal{F}_{\Gamma^*}([\alpha,\phi]),\\
\big\lbrace\mathcal{F}_{\Gamma^*}(\phi),\mathcal{F}_{\Gamma^*}(\phi')\big\rbrace &=& 0,
\ea
\end{subequations}
we see that the Hamiltonians (\ref{curvature constraint}) and (\ref{gauss hamiltonian}) form a first class algebra.

We are interested in the phase space obtained from $\mathcal{P}$ by symplectic reduction with respect to $\mathcal{F}_{\Gamma^*}$ and $\mathcal{G}_\Gamma$, which we denote by
\be
\mathcal{P}_{\Gamma, \Gamma^*}\equiv T^*\mathcal{A}\sslash\big(\mathcal{F}_{\Gamma^*}\times\mathcal{G}_\Gamma\big)=\mathcal{C}/\big(\mathcal{F}_{\Gamma^*}\times\mathcal{G}_\Gamma\big),
\ee
where
\be\label{C}
\mathcal{C}\equiv\big\lbrace(A,E)\in T^*\mathcal{A}|F(A)(x)=\rd_AE(y)=0,\:\forall\:x\in\Sigma\backslash\Gamma^*,\:\forall\:y\in\Sigma\backslash V_\Gamma\big\rbrace
\ee
is the constrained space. This is the infinite-dimensional space of flat $\SU(2)$ connections on $\tilde{\Sigma}\equiv\Sigma\backslash\Gamma^*$, and fluxes satisfying the Gauss law outside of $V_\Gamma$. Once we divide this constrained space by the action of the two gauge groups introduced above, we obtain the finite-dimensional orbit space $\mathcal{P}_{\Gamma, \Gamma^*}$ \cite{AandB}. We are going to prove that $\mathcal{P}_{\Gamma, \Gamma^*}$ is the continuous analogue of the discrete spin network phase space $P_\Gamma$.

Let us start by constructing a three-dimensional cellular decomposition of the region. Since we have chosen $\Gamma^*$ to be the one-skeleton $\Delta_1$ of the cellular decomposition $\Delta$ of $\Sigma$, the cellular decomposition of $\tilde{\Sigma}$ is simply given by $\tilde{\Delta}\equiv\Delta\backslash\Delta_1$. Explicitly, the decomposition $\tilde{\Delta}$ can be written as
\be
\tilde{\Delta}=\bigcup_vC_v\bigcup_eF_e,
\ee
where $C_v$ are three-dimensional open cells labeled by the vertices $v\in\Gamma$, and $F_e$ are two-dimensional open cells labeled by the edges $e\in\Gamma$. We would like to solve the curvature constraint $F(A)=\rd_AA=0$ on $\tilde{\Sigma}$ and the Gauss constraint $\rd_AE=0$ on $\Sigma\backslash V_\Gamma$. We start by solving them for each three-dimensional cell $C_v$.

To solve the curvature constraint, let us define on a three-cell $C_v$ a group-valued map $a_v(x):C_v\longrightarrow\SU(2)$ as the path-ordered exponential
\be
a_v(x)\equiv\Pexp\int_x^vA,
\ee
where the integration can be taken over any arbitrary path from the point $x\in C_v$ to the vertex $v$ because the connection is flat and $C_v$ is simply connected. By construction, this map is such that $a_v(v)=1$. This allows us to reconstruct on $C_v$ the flat connection $A$ as
\be\label{flat connection}
A(x)=a_v(x)\rd a_v^{-1}(x).
\ee
The second constraint to satisfy is the Gauss law outside of the vertex $v$ which lies inside the cell $C_v$. Because the connection is flat, the covariant derivative of the electric field $E$ can be written as
\be
\rd_AE=dE+\left[a_v\rd a_v^{-1},E\right]=a_v\rd\left(a_v^{-1}Ea_v\right)a_v^{-1}=a_v\rd X_va_v^{-1},
\ee
where we have introduced the Lie algebra-valued two-form field
\be\label{two-form}
X_v(x)\equiv a_v(x)^{-1}E(x)a_v(x).
\ee
Therefore, we see that the Gauss law implies that the two-form $X_v$ is closed outside of $v$ since
\be\label{gauss Xv}
\rd X_v(x)=a_v(x)^{-1}\rd_AE(x)a_v(x)=0,\qquad\qquad\forall\:x\in C_v-\{v\}.
\ee
The electric field can now easily be reconstructed since we have
\be\label{flux field}
E(x)=a_v(x)X_v(x)a_v(x)^{-1}.
\ee
One can conclude that a general solution of the two constraints $F(A)=\rd_AA=0$ and $\rd_AE=0$ on $C_v$ and $C_v-\{v\}$ respectively, is given in terms of a Lie algebra-valued closed two-form $X_v$ and a group element $a_v:C_v\longrightarrow\SU(2)$, the connection and flux fields being given by (\ref{flat connection}) and (\ref{flux field}).

Now we can extend this solution to the whole space $\tilde{\Sigma}$ by gluing consistently the solutions on each cell. We have labeled the three-dimensional cells $C_v$ with vertices of the  graph $\Gamma$. Consequently, the two-dimensional cells $F_e$, labeled by edges $e=(v_1v_2)$ of $\Gamma$ connecting two vertices (such that $s(e)=v_1$ and $t(e)=v_2$), are obtained by intersecting two three-dimensional cells as
\be
F_e=\overline{C_{v_1}}\cap\overline{C_{v_2}},
\ee
where the bar denotes the closure of the cell. We assume that the two-dimensional cells $F_e$ are oriented, and that their orientation is reversed when we change the orientation of the edge $e$.  Demanding that the connection and flux fields be continuous across the two-dimensional cells amounts to assuming that there exists, for each $F_e$, an $\SU(2)$ element $h_e$ such that
\be\label{conditions}
a_{v_2}(x)=a_{v_1}(x)h_e,\qquad\qquad X_{v_2}(x)=h_e^{-1}X_{v_1}(x)h_e,
\ee
for $x\in F_e$. Notice that the first equality can be written as
\be
\label{def a}
h_e(A)=a_{s(e)}(x)^{-1}a_{t(e)}(x)=\Pexp\int_eA,
\ee
where $x$ is any point on the two-cell $F_e$, and once again the definition does not depend on $x$ because the connection is flat. By construction, one can see that under an orientation reversal we have $h_{e^{-1}}=h_e^{-1}$.

This construction shows that the constrained space $\mathcal{C}$ is isomorphic to the data $(a_v,X_v,h_e)$, subject to the conditions (\ref{conditions}). We are now interested in the quotient of this constrained space by the gauge group $\mathcal{F}_{\Gamma^*}\times\mathcal{G}_\Gamma$. Elements of this gauge group are pairs $\big(\phi(x),g_{\!o}(x)\big)$, where $\phi$ is a Lie algebra-valued one-form which vanishes on $\Gamma^*$, and $g_{\!o}$ is an element of $\SU(2)$ (obtained by exponentiation of $\alpha$) fixed to the identity of the group at the vertices $V_\Gamma$. The action of $\mathcal{F}_{\Gamma^*}\times\mathcal{G}_\Gamma$ on the pair $(A,E)\in\mathcal{P}$ translates on the constraint surface $\mathcal{C}$ into an action on the data $(a_v,X_v,h_e)$ given by
\be
a_v(x)\longrightarrow g_{\!o}(x)a_v(x),\qquad X_v(x)\longrightarrow X_v(x)+\rd\left(a_v(x)^{-1}\phi(x)a_v(x)\right),\qquad h_e\longrightarrow h_e.
\ee
Following (\ref{new flux}), let us compute the flux $X_e$ across a surface dual to an edge $e$ which is such that $s(e)=v$. It is given by
\be\label{Xe and Xv relation}
X_e=\int_{F_e}h_{\pi_e}(x)E(x)h_{\pi_e}(x)^{-1}=\int_{F_e}a_v(v)a_v(x)^{-1}E(x)a_v(x)a_v(v)^{-1}=\int_{F_e}X_v,
\ee
where we have used the fact that $a_v(v)=1$. We see that the observables which are invariant under this gauge transformation are simply given by the holonomies $h_e$ and the fluxes $X_e$.

\subsection{The symplectomorphism between $\mathcal{P}_{\Gamma, \Gamma^*}$ and $P_\Gamma$}

\noindent Now we come to our main result, which is the symplectomorphism between the continuous phase space $\mathcal{P}_{\Gamma, \Gamma^*}$ and the discrete spin network phase space $P_\Gamma$. Let us construct a map between the constrained continuous data in $\mathcal{C}$ (see (\ref{C})) and discrete data on the spin network phase space $P_\Gamma$, and denote it by
\be
\begin{array}{cccl}
\mathcal{I}:&\mathcal{C}&\longrightarrow&P_\Gamma\\
&(A,E)&\longmapsto&\big(h_e(A),X_e(A,E)\big).
\end{array}
\ee
For this, we define for every three-cell $C_v$ a group-valued map $a_v:C_v\longrightarrow\SU(2)$ such that $a_v(v)=1$ and a Lie algebra-valued two-form $X_v:C_v\longrightarrow\Omega^2\big(C_v,\su(2)\big)$ closed outside of the vertices of $\Gamma$.
 Given these fields, we can reconstruct on $C_v$ the connection and the two-form field using
\be
A(x)=a_v(x)\rd a_v(x)^{-1},\qquad\qquad E(x)=a_v(x)X_v(x)a_v(x)^{-1}.
\ee
The map $\mathcal{I}$ is then defined by
\begin{subequations}\label{map}
\ba
h_e(A)&\equiv&\Pexp\int_eA
=a_{s(e)}(x)^{-1}a_{t(e)}(x),\label{map1}\\
X_e(A,E)&\equiv&\int_{F_e}h_{\pi_e}(x)E(x)h_{\pi_e}(x)^{-1}
=\int_{F_e}X_{s(e)}(x),\label{map2}
\ea
\end{subequations}
where in the definition of $h_e$, $x$ is any point on the two-cell $F_e$, and once again the definition does not depend on $x$ because the connection is flat. To compute the holonomy $h_e$, we have used the group elements $a_{s(e)}(x)$ and $a_{t(e)}(x)$ to define the connection on the two cells dual to the vertices $s(e)$ and $t(e)$ respectively.

It is possible to use equation (\ref{map2}) to write down the relationship between the discrete and continuous Gauss laws. We already know from (\ref{gauss Xv}) that the Gauss law is equivalent to the requirement that the two-form $X_v$ be closed outside of the vertex $v$. We can now write that
\be
\int_{C_v}a_v(x)^{-1}\rd_AE(x)a_v(x)=\int_{C_v}\rd X_v=\int_{\cup_eF_e=\partial C_v}X_{s(e)}=\sum_{e|s(e)=v}X_e=G_v,
\ee
which relates the continuous and discrete constraints. This shows that the violation of the continuous Gauss constraint is located at the vertices of $\Gamma$, and given by a distribution determined by the discrete Gauss constraint:
\be\label{pointGauss}
\rd_AE(x)=\sum_{v\in V_{\Gamma}}G_{v}\delta(x-v).
\ee

Since the map $\mathcal{I}$ is invariant under the gauge transformations $\mathcal{F}_{\Gamma^{*}}\times \mathcal{G}_{\Gamma}$ we can write it as a map
\be\nonumber
[\mathcal{I}]:\mathcal{P}_{\Gamma,\Gamma^*}\longrightarrow P_\Gamma.
\ee
We will now show that this map is not only invertible, but also a symplectomorphism.

\begin{proposition}\label{proposition}
The map ${[}\mathcal{I}{]}:\mathcal{P}_{\Gamma, \Gamma^*}\longrightarrow P_\Gamma$ defined by (\ref{map}) is a symplectomorphism, and is invariant under the action of diffeomorphisms connected to the identity preserving $\Gamma^*$ and the set $V_\Gamma$ of vertices of $\Gamma$.
\end{proposition}

We are going to prove this proposition in the remainder of this work. Before doing so, let us stress that this result implies the existence of an inverse map which allows one to reconstruct from the discrete data an equivalence class $[A(h_e),E(h_e,X_e)]$ of continuous configurations satisfying the curvature and Gauss constraints (i.e. configurations in the constrained space $\mathcal{C}$). Explicitly, this equivalence class is defined with respect to the equivalence relation
\be\label{equivalence class}
(A,E)\sim\left(g_{\!o}\triangleright A,g_{\!o}^{-1}(E+\rd_A\phi)g_{\!o}\right),
\ee
where once again $\phi$ is a Lie algebra-valued one-form vanishing on $\Gamma^*$, and $g_{\!o}$ is an element of $\SU(2)$ fixed to the identity of the group at the vertices $V_\Gamma$.

Evidently, Proposition \ref{proposition} implies a similar proposition for the gauge-invariant phase spaces. Indeed, if one defines
\be
\mathcal{P}_{\Gamma,\Gamma^*}^\mathcal{G}\equiv T^*\mathcal{A}\sslash\big(\mathcal{F}_{\Gamma^*}\times\mathcal{G}\big)=\mathcal{C}^\mathcal{G}/\big(\mathcal{F}_{\Gamma^*}\times\mathcal{G}\big),
\ee
where
\be
\mathcal{C}^\mathcal{G}\equiv\big\lbrace(A,E)\in T^*\mathcal{A}|F(A)(x)=\rd_AE(y)=0,\:\forall\:x\in\Sigma\backslash\Gamma^*,\:\forall\:y\in\Sigma\big\rbrace,
\ee
and $\mathcal{G}=C^\infty\big(\Sigma,\SU(2)\big)$ is the group of full $\SU(2)$ gauge transformations, we have the symplectomorphism $\mathcal{P}_{\Gamma, \Gamma^*}^\mathcal{G}=P^G_\Gamma$ between the continuous and discrete gauge-invariant phase spaces. This follows directly from Proposition \ref{proposition}, and the fact that $\mathcal{G}=\mathcal{G}_\Gamma\times G_\Gamma$, where $G_\Gamma$ is the group of discrete gauge transformations acting at the vertices $v\in V_\Gamma$ only.

Notice that when we act with the full group $\mathcal{G}$ of $\SU(2)$ transformations, the holonomies $h_e$ and the fluxes $X_e$ clearly become gauge-covariant, i.e. satisfies $\mathcal{I}(g\triangleright A,g\triangleright E)=g\triangleright\mathcal{I}(A,E)$. Indeed, since the group element $g$ is not fixed to the identity at the vertices $v$ anymore, we have $g\triangleright a_v(x)=g(x)a_v(x)g(v)^{-1}$, and therefore the definition (\ref{conditions}) tells us that we have $g\triangleright h_e=g_{v_1}h_eg_{v_2}^{-1}$, where $e$ is an edge of $\Gamma$ connecting the vertices $v_1$ and $v_2$.

\subsection{The symplectic structures}

\noindent In this subsection we use the map (\ref{map}) to prove the equivalence of the symplectic structures on the continuous and discrete spaces $\mathcal{P}_{\Gamma, \Gamma^*}$ and $P_\Gamma$. We know that the spaces $\mathcal{P}$ and $P_\Gamma$ are symplectic manifolds, their symplectic structures being given by (\ref{symplectic gravity}) and (\ref{symplectic}) respectively. Since the space $\mathcal{P}_{\Gamma, \Gamma^*}$ has been obtained from $\mathcal{P}$ by symplectic reduction, the Marsden-Weinstein theorem ensures that it also carries a symplectic structure. We are now going to show that the symplectic structures on the spaces $\mathcal{P}_{\Gamma, \Gamma^*}$ and $P_\Gamma$ are in fact identical.

Let us start with the symplectic potential coming from the first order formulation of gravity. It is given by
\ba
\Theta=\f{1}{2}\int_\Sigma\Tr\left(\star(e\wedge e)\wedge\delta A\right)=\int_\Sigma\Tr\left(E\wedge\delta A\right),
\ea
where $\star$ denotes the Hodge duality map in the Lie algebra $\su(2)$. We first use the cellular decomposition $\Delta$ to evaluate this symplectic potential on the set of partially flat connections and write
\begin{subequations}
\ba
\Theta
&=&\sum_v\int_{C_v}\Tr\left(E\wedge\delta\left(a_v\rd a_v^{-1}\right)\right) \\
&=&\sum_v\int_{C_v}\Tr\left(X_v\wedge\rd\left(\delta a_v^{-1}a_v\right)\right) \\
&=&\sum_v\int_{\partial C_v}\Tr\left(X_v\delta a_v^{-1}a_v\right)-\sum_{v}G_{v}\delta a_v^{-1}a_v(v)\\\
&=&\sum_v\int_{\partial C_v}\Tr\left(X_v\delta a_v^{-1}a_v\right),
\ea
\end{subequations}
where we have used the identity $\delta\left(a_v\rd a_v^{-1}\right)=a_v\rd\left(\delta a_v^{-1}a_v\right)a_v^{-1}$, the definition (\ref{two-form}) of the two-form field $X_v$, and the fact that $\rd X_{v}= G_{v}\delta(x-v)$ (see equation(\ref{pointGauss})). The last equality follows from the condition $a_{v}(v)=1$, which implies $\delta a_{v}(v)=0$. The summation over three-cells dual to the vertices $v$ can be rearranged as a sum over two-cells dual to the edges $e$, which gives
\be
\Theta=\sum_e\int_{F_e}\left[\Tr\left(X_{s(e)}\delta a_{s(e)}^{-1}a_{s(e)}\right)-\Tr\left(X_{t(e)}\delta a_{t(e)}^{-1}a_{t(e)}\right)\right].
\ee
Now we can use the condition (\ref{conditions}) of compatibility of the group elements $a_v$ across the edges to rewrite the second term and obtain
\be
\Theta=\sum_e\int_{F_e}\left[\Tr\left(X_{s(e)}\delta a_{s(e)}^{-1}a_{s(e)}\right)-\Tr\left(h_e^{-1}X_{s(e)}h_e\delta\left(h_e^{-1}a_{s(e)}^{-1}\right)a_{s(e)}h_e\right)\right].
\ee
Finally, we can expand the last term to find the result
\be
\Theta=-\sum_e\int_{F_e}\Tr\left(h_e^{-1}X_{s(e)}h_e\delta h_e^{-1}h_e\right)
=\sum_e\Tr\left(X_e\delta h_eh_e^{-1}\right).
\ee
This is exactly the symplectic potential associated to $|E_\Gamma|$ copies of the cotangent bundle $T^*\SU(2)$. It shows that the symplectic structure of the spin network phase space is equivalent to that of first order gravity evaluated on the set of partially flat connections. In particular, since the symplectic forms are invertible by definition, this proves that the continuous phase space $\mathcal{P}_{\Gamma,\Gamma^*}$ is indeed finite-dimensional and isomorphic to $P_{\Gamma}$.

\subsection{Action of diffeomorphisms}
\label{diff-gauge}

\noindent Now we prove the second point of Proposition \ref{proposition}, which concerns the invariance of the symplectomorphism under a certain class of diffeomorphisms. The isomorphism $\mathcal{I}:\mathcal{P}_{\Gamma, \Gamma^*}\longrightarrow P_\Gamma$ depends on a choice of cellular decomposition $\Delta$ dual to $\Gamma$ with one-skeleton $\Delta_1=\Gamma^*$. Diffeomorphisms $\Phi\in\text{Diff}(\Sigma)$ act naturally on the continuous phase space $\mathcal{P}_{\Gamma, \Gamma^*}$ by $A\longmapsto\Phi^*A$ and $E\longmapsto\Phi^*E$.

Let us start by choosing a particular diffeomorphism $\Phi_{\!o}$ which preserves the graph $\Gamma^*$ and the vertices $V_\Gamma$ inside the cells $C_v$, and is connected to the identity\footnote{This means that there exists a smooth one-parameter family of diffeomorphism $\Phi_{t}$ such that $\Phi_{t=0}=\text{id}$ and $\Phi_{t=1}=\Phi_{\!o}$.}. Because the connection is flat on $\tilde{\Sigma}$, the holonomy $h_e(A)$ is independent of the choice of path between $s(e)$ and $t(e)$ as long as any two paths are in the same homotopy class of $\tilde{\Sigma}$. The edges $e$ and $\Phi_{\!o}(e)$ are in the same homotopy class if $\Phi_{\!o}$ is connected to the identity and not moving $\Gamma^{*}$. Then it is clear that we have
\be\label{diffeo1}
h_e(\Phi_{\!o}^*A)=h_{\Phi_{\!o}(e)}(A)=h_e(A).
\ee
Similarly, the action of $\Phi_{\!o}$ on the group element $a_v(x)$ maps it to $a_v\big(\Phi_{\!o}(x)\big)$. This implies that the two-form $X_v$ defined by (\ref{two-form}) satisfies $X_v\big(\Phi_{\!o}(x)\big)=\Phi_{\!o}^*X_v(x)$. Recall from the definitions of the cellular decomposition that each face $F_e$ is bounded by links in the one-skeleton $\Gamma^*$. Now, since $\Phi_{\!o}$ does not move the graph $\Gamma^*$, we have that $\partial F_e=\partial\big(\Phi_{\!o}(F_e)\big)\in\Gamma^*$, and therefore $F_e\cup\Phi_{\!o}(F_e)$ encloses a volume, which furthermore does not contain any vertices of $\Gamma$. Thus, by virtue of (\ref{same flux}) and (\ref{Xe and Xv relation}), we have that
\be\label{diffeo2}
X_e(\Phi_{\!o}^*A,\Phi_{\!o}^*E)=X_e(A,E).
\ee
Relations (\ref{diffeo1}) and (\ref{diffeo2}) together show that $\mathcal{I}\circ\Phi_{\!o}=\mathcal{I}$.

We can give another very elegant proof of the invariance of the map $\mathcal{I}$ under the diffeomorphisms $\Phi_{\!o}$. For this, recall that given a vector field $\xi^a$, a diffeomorphism acts on the connection like
\be
\mathcal{L}_\xi A=\rd(\iota_\xi A)+\iota_\xi\rd A=\iota_\xi F+\rd_A(\iota_\xi A),
\ee
and on the electric field like
\be
\mathcal{L}_\xi E=\rd(\iota_\xi E)+\iota_\xi\rd E=\iota_\xi\rd_AE+\rd_A(\iota_\xi E)+[E,\iota_\xi A],
\ee
where $\iota$ denotes the interior product. Now, if the data $(A,E)$ is on the constraint surface $\mathcal{C}$, the curvature vanishes outside of $\Gamma^*$, while $\rd_AE$ vanishes outside of the set $V_\Gamma$ of vertices. Therefore, if we consider a vector field $\xi^a$ which vanishes on $\Gamma^*$ and on $V_\Gamma$, we see that the action of diffeomorphisms is a combination of flat transformations (\ref{flat-gauge}) and gauge transformations (\ref{SU(2)-gauge}) with field-dependent parameters of transformation, i.e.
\be
\mathcal{L}_\xi A=\delta^{\mathcal{F}_{\Gamma^*}}_{\iota_\xi E}A+\delta^{\mathcal{G}_\Gamma}_{\iota_\xi A}A,\qquad\qquad
\mathcal{L}_\xi E=\delta^{\mathcal{F}_{\Gamma^*}}_{\iota_\xi E}E+\delta^{\mathcal{G}_\Gamma}_{\iota_\xi A}E.
\ee
We can write this more succinctly as simply
\be
\mathcal{L}_\xi=\delta^{\mathcal{F}_{\Gamma^*}}_{\iota_\xi E}+\delta^{\mathcal{G}_\Gamma}_{\iota_\xi A}.
\ee
Now, since the holonomy and flux variables are invariant under the flatness and gauge transformations, such diffeomorphisms vanish on the variables $(h_e,X_e)$.

\section{Gauge choices for the electric field}
\label{sec4}

\noindent Now that we have established the isomorphism between $P_\Gamma$ and the continuous phase space $\mathcal{P}_{\Gamma,\Gamma^*}$, we have a correspondence between discrete geometries and an equivalence class of continuous geometries related according to (\ref{equivalence class}) by group gauge transformations and translations. Up to group gauge transformations, the holonomy uniquely determines a choice of connection. For the E field, however, the story is different since even after we have performed a group gauge transformation, there is still a huge ambiguity coming from the transformation $E \rightarrow E+ \rd_A \phi$ on the continuous electric field determined by the fluxes. It is clear that in order to construct a continuous field configuration starting from the discrete data, one has to specify which continuous field representative to pick in the particular equivalence class determined by the discrete data. In other words, a choice of a representative in this equivalence class is a choice of gauge. More precisely, we have the following definition:
\begin{definition}
A choice of gauge is a map from the discrete data to the continuous phase space,
\be\begin{array}{cccl}
\mathcal{T}:&P_\Gamma &\longrightarrow&\mathcal{C}\\
&(h_e,X_e)&\longmapsto&(A,E),
\end{array}\ee
which is the inverse of $\mathcal{I}$ in the sense that
\be\label{inverse map}
\mathcal{I}\circ\mathcal{T}=\mathrm{id}.
\ee
We say that a gauge fixing is diffeomorphism-covariant if $\Phi^*\mathcal{T}$ is equal to the map $\mathcal{T}$  defined on the graphs $\Phi(\Gamma)$ and $\Phi(\Gamma^*)$, for any diffeomorphism $\Phi:\Sigma\longrightarrow\Sigma$.
\end{definition}

In other words, choosing a gauge amounts to giving a prescription for  reconstructing continuous fields $A(h_e)$ and $E(X_e,h_e)$ starting from the discrete data, such that (\ref{inverse map}) holds, i.e.
\be
h_e\big(A(h_e)\big)=h_e,\qquad\qquad X_e\big(A(h_e),E(X_e,h_e)\big)=X_e.
\ee
Note that a gauge fixing $\mathcal{T}$ is a right inverse for $\mathcal{I}$, while the reverse is not true. The map $\mathcal{T} \circ\mathcal{I}$ is not the identity, it just maps a continuous configuration $(A,E)$ that solves the Gauss and curvature constraints into another gauge-equivalent configuration which satisfies the gauge choice.

As we have already seen, at the continuous level a flat connection on $\tilde{\Sigma}$ is determined on every cell $C_v$ by a group element $a_v(x)$. Locally, it is always possible to perform a gauge transformation that sends this element to the identity of the group, and thereby construct a trivial connection. If we pick two neighboring cells $C_{v_1}$ and $C_{v_2}$ such that the vertices $v_1$ and $v_2$ bound the edge dual to the face $F_e=\overline{C_{v_1}}\cap\overline{C_{v_2}}$, the relevant gauge-invariant information about the connection is encoded in the transition group element $h_e$.

For the electric field, there is more gauge freedom since the variable $E$ can be acted upon by both $\mathcal{F}_{\Gamma^*}$ and $\mathcal{G}_\Gamma$. Therefore, there is a priori a huge ambiguity in the choice of gauges that one can choose to reconstruct the continuous data. This means that knowledge of the fluxes does not accurately determine the geometry of space, but only a family of geometries that are gauge-equivalent under translations of the type $E\longmapsto E+\rd_A\phi$.

However, there is a powerful way in which we can restrict the gauge choices that are available. This can be done by asking that a gauge choice transforms covariantly under the action of diffeomorphisms. A diffeomorphism $\Phi$ of $\Sigma$ acts on the continuous data in the usual manner $(A,E)\longmapsto\big(\Phi^*A,\Phi^*E\big)$. The same diffeomorphism also acts on the discrete data $(h_e, X_{F_e})$ as $\big(h_{\Phi(e)},X_{\Phi(F_e)}\big)$. Note that here we have made explicit the fact that the flux field $X_e$ depends on $\Gamma^*$ via the choice of a surface $F_e$ whose boundary is supported on $\Gamma^*$. A gauge choice is said to be covariant if this action of the diffeomorphisms commutes with the gauge map $\mathcal{T}$.

If we restrict ourselves to gauge choices that are covariant under the action of diffeomorphisms, the ambiguity in the gauge choices is dramatically resolved, and there are only a few choices available. In the following we present two such gauge choices\footnote{We conjecture these are the only two possible gauge choices, but a detailed investigation of this is still needed.}. First, the singular gauge choice in which the electric field $E$ vanishes outside of $\Gamma$, and then the flat gauge in which $E$ is flat outside of $\Gamma^{*}$. It is remarkable that these two gauge choices correspond to the two main interpretations of the fluxes used in the literature. In loop quantum gravity one usually interprets the $E$ field as having support only on $\Gamma$ since the corresponding operator acting on a spin network state gives $\widehat{E}(x)\left|\psi\right>=0$ for $x \not\in\Gamma$. On the other hand, the spin foam literature usually interprets the $E$ field as being flat outside of $\Gamma^{*}$. Our analysis shows that these two pictures are not contradictory, but that they correspond to two different covariant gauge choices underlying the same discrete data.

Now we want to emphasize that the restriction on the gauge choices coming from the requirement of covariance under diffeomorphisms is the analog of the so-called uniqueness theorem of the quantum representation of the holonomy-flux algebra \cite{LOST}. This theorem states that there is a unique diffeomorphism-covariant gauge choice, which corresponds to the singular gauge in which $E$ has support on the graph $\Gamma$ only and vanishes on $\Gamma^*$. In this singular gauge, which we refer to as the LQG gauge, the electric field $E$ vanishes outside of the graph $\Gamma$ dual to the triangulation $\Delta$. This can be written as $E|0\rangle=0$, where the vacuum state $|0\rangle$ is the state of no geometry. Indeed, in LQG excitations of quantum geometry have support on the graph $\Gamma$ only. Therefore, in all the regions of $\Delta$ outside of $\Gamma$, there is simply no geometry, and the electric field vanishes. We are going to give below an explicit construction of the continuous singular electric field.

The key observation is that there is another legitimate choice of representative configuration in the equivalence class (\ref{equivalence class}) of continuous geometries which respects the diffeomorphism symmetry. As we already said, it is given by the flat gauge. At the quantum level, this corresponds to a choice of a vacuum state $|0_\text{\tiny{F}}\rangle$ in which the curvature vanishes. This corresponds to the flat, or spin foam gauge, in which we have  $F(A)|0_\text{\tiny{F}}\rangle=0$. This diffeomorphism-invariant vacuum is different from the one singled out by the LOST theorem \cite{LOST} (which obviously corresponds to the singular gauge) and it would be interesting to investigate further its properties. What should be noted is that such a vacuum state appears naturally in our context and that it corresponds to the spin foam description. It can be seen as the dual of the singular gauge, in the sense that it defines a flat geometry within the cells $C_v$, with a non-vanishing electric field $E$ on the dual graph $\Gamma^*$. As we will see in more detail, the availability of this gauge clearly shows that it is possible to define a locally flat geometry without necessarily having a triangulation with straight edges and flat faces. In Regge geometries \cite{regge}, the extrinsic curvature is concentrated along the one-skeleton $\Delta_1$ of the triangulation, but in the present construction, the edges of $\Gamma^*$ are not necessarily straight.

Here we have drawn a parallel between a choice of gauge at the classical level and a choice of a vacuum state at the quantum level. It would be interesting to develop this analogy further.

In the remainder of this section, we are going to study in more detail the singular and flat gauges for the electric field. Our goal is to study the gauge freedom for the basic variables on the continuous phase space, and to construct explicitly the electric field as a functional of the discrete variables $h_e$ and $X_e$.

\subsection{Singular gauge}

\noindent The singular gauge is a gauge in which the electric field $E$ vanishes outside of the graph $\Gamma$. In this section, we show by an explicit construction that it always possible to make such a gauge choice. More precisely, we construct explicitly continuous fields $A(h_e)$ and $E_\text{\tiny{S}}(h_e,X_e)$ which are such that $E_\text{\tiny{S}}(x)=0$ if $x\notin\Gamma$, and which satisfy the property $\mathcal{I}(A,E_\text{\tiny{S}})=(h_e,X_e)$ under the action of the map (\ref{map}).

In order to prove this, let us first introduce the following form:
\be
\omega(x,y)\equiv\omega^i(x-y)\epsilon_{ijk}\rd x^j\wedge\rd y^k,\qquad\qquad\text{with}\qquad\omega^i(x)\equiv\f{1}{4\pi}\f{x^i}{|x|^3}.
\ee
This object is a (1,1)-form, i.e. a one-form in $x$, and a one-form in $y$. This form satisfies a key property, which is summarized in the following lemma.
\begin{lemma}
There exists an $\alpha(x,y)$ which is a (2,0)-form (i.e. a two-form in $x$ and a zero-form in $y$), such that
\be
\rd_x\omega(x,y)+\rd_y\alpha(x,y)=\delta(x,y),
\ee
where $\rd_x\equiv\rd x^i\partial_{x^i}$, and $\delta(x,y)$ is the distributional (2,1)-form
\be
\delta(x,y)=\delta(x-y)\epsilon_{ijk}\rd x^i\wedge\rd x^j\wedge\rd y^k
\ee
vanishing outside of $x=y$.
\end{lemma}
\begin{proof}
First, it is straightforward to show that $\partial_i\omega^i(x) =0$ for $x\neq 0$. Moreover, it is possible to show by a direct computation in spherical coordinates that
\be
\int_{S_{\varepsilon}}\omega^i(x)\epsilon_{ijk}\rd x^j\wedge\rd x^k=2,
\ee
where $S_{\varepsilon}$ is a sphere of radius $\varepsilon$.
Since this integral is also equal to
\be
2\int_{B_{\varepsilon} }\partial_i\omega^i(x)\rd^3x,
\ee
where $B_{\varepsilon}$ is the ball of radius $\varepsilon$, we obtain that $\partial_i\omega^i(x)=\delta(x)$. By a direct computation we can now get that
\be
\rd_x\omega(x,y)=s_k\wedge\rd y^k\left(\partial_i\omega^i\right)(x-y)-s_j\wedge\rd y^k\left(\partial_k\omega^j\right)(x-y),
\ee
with
\be
s_i=\f{1}{8\pi}\epsilon_{ijk}\rd x^j\wedge\rd x^k.
\ee
The lemma is therefore established by introducing $\alpha(x,y)\equiv\omega^i(x-y) s_i$.
\end{proof}

Given this lemma, it is now a straightforward task to construct a singular flux field. For this, we first construct a flat connection $A$ on $\tilde{\Sigma}$ following the construction of subsection \ref{subsection flat}, and then  we define the singular flux field as
\be
E_\text{\tiny{S}}(x)\equiv\rd_{A}\left(\sum_{e\in\Gamma} h_{\pi_e}(x)^{-1}X_{e}h_{\pi_e}(x)\int_{e(y)}\omega(x,y)\right).
\ee
The integral entering this definition is a one-dimensional integral over the edge $e$ parametrized by the variable $y$, which implies that the term inside the parenthesis is a one-form in $x$.

The proof that this flux satisfies all the desired requirements is straightforward. First, it is obvious that the Gauss law $\rd_AE_\text{\tiny{S}}=0$ is satisfied on $\Sigma\backslash\Gamma^{*}$ since $\rd_{A}^{2}=F(A)=0$ on this space.
Moreover, using the previous lemma and the definition of the holonomy, we can compute explicitly the covariant derivative:
\be\label{expanded singular E}
E_\text{\tiny{S}}(x)=\sum_eh_{\pi_e}(x)^{-1}X_eh_{\pi_e}(x)\left(\vphantom{\f{1}{1}}\delta_e(x)-\alpha\big(x,s(e)\big)+\alpha\big(x,t(e)\big)\right),
\ee
where
\be
\delta_e(x)\equiv\int_{e(y)}\delta(x,y).
\ee
The last two terms in (\ref{expanded singular E}) can be reorganized in terms associated with the vertices to find
\be
E_\text{\tiny{S}}(x)=\sum_eh_{\pi_e}(x)^{-1}X_eh_{\pi_e}(x)\delta_e(x)-\sum_v\alpha(x,v) h_v(x)^{-1}\left(\sum_{e|s(e)=v}X_e-\sum_{e|t(e)=v}h_e^{-1}X_eh_e\right)h_v(x),
\ee
where $h_v(x)$ is the holonomy going from the vertex $v$ to the point $x$. Now the last term vanishes due to the discrete Gauss law (\ref{discrete gauss}). Therefore, we finally find that the singular electric field is
\be
\label{Esing}
E_\text{\tiny{S}}(x)=\sum_eh_{\pi_e}(x)^{-1}X_eh_{\pi_e}(x)\delta_e(x).
\ee
This electric field is obviously vanishing outside of $\Gamma$, and is such that $X_e(A,E_\text{\tiny{S}})=X_e$. It is interesting to note that the integral of the two-form $\alpha(x,y)$ along $S$,
\be
\int_S\alpha(x,y)=\f{1}{8\pi}\int_S\omega^i(x-y)\epsilon_{ijk}\rd x^j\wedge\rd x^k,
\ee
is simply the solid angle of $S$ as viewed from $y$ divided by $4\pi$.

\subsection{Flat cell gauge}

\noindent The flat cell gauge is a choice of electric field with vanishing intrinsic and extrinsic curvature within the cells, i.e. with the scalar curvature $R\equiv\rd_\Gamma\Gamma=0$ and the extrinsic curvature $K=0$ in each cell $C_v$. Note that $R$ and $K$ are not necessarily zero on the faces $F_e$ and their boundaries $\partial F_e$. This gauge choice requires that we be within the $\SU(2)$-invariant phase space.

We are about to prove that it is always possible to find a gauge transformation, generated by the flatness constraint, which takes an arbitrary electric field $E\in\mathcal{P}^\mathcal{G}_{\Gamma,\Gamma^*}$ to a flat electric field $\bar{E}$. In the following we assume the frame field $e$ is invertible.

Let us begin with two lemmas:
\begin{lemma}\label{K zero}
Extrinsic curvature is zero if and only if the frame field is torsion-free.
\end{lemma}
\begin{proof}
Torsion is given by
\be\label{torsion-free}
\rd_Ae=\rd e+\left[\Gamma+\gamma K,e\right]=\gamma\left[K,e\right],
\ee
where by definition, the spin connection $\Gamma$ is the solution to $\rd e+\left[\Gamma,e\right]=0$. This equation shows that $K=0$ implies $\rd_Ae=0$. To show that the reverse is also true, we use (\ref{torsion-free}) in index-form to write
\be\label{tf1}
\frac{1}{\gamma}\epsilon^{abc}D_be_{ic}=\epsilon^{abc}\epsilon_{ijk}K_b^je_c^k=(\det e)K^j_b(e^a_i e^b_j-e^a_j e^b_i),
\ee
where $D_a$ is the covariant exterior derivative in index notation and we used the identity $\epsilon^{abc}\epsilon_{ijk}e^k_c\equiv(\det e)(e^a_ie^b_j-e^a_je^b_i)$ in the second equality. Contracting both sides of this equation with $e_a^i$ leads to an equation on the trace of the extrinsic curvature:
\be\label{tf2}
K_a^a=\frac{1}{2\gamma(\det e)}\epsilon^{abc}e_a^iD_be_{ci}.
\ee
Now using (\ref{tf2}) in (\ref{tf1}), we find
\be\label{extrinsic}
K^i_a=\frac{1}{\gamma(\det e)}\left(\frac{1}{2}e^i_ae_d^j-e^i_de^j_a\right)\epsilon^{dbc}D_be_{cj}.
\ee
This shows that $\rd_Ae=0$ implies $K=0$ and establishes the proof.
\end{proof}

\begin{lemma}
A flat connection, together with vanishing extrinsic curvature, imply that intrinsic curvature is zero.
\end{lemma}
\begin{proof}
Using the definition $A\equiv\Gamma+\gamma K$ of the Ashtekar-Barbero connection, we can write its curvature as
\be
F(A)=R+\gamma\rd_\Gamma K+\frac{\gamma^2}{2}[K,K].
\ee
Setting $F(A)=0$ and $K=0$ implies that $R=0$.
\end{proof}

We showed previously (see (\ref{flat connection}) and (\ref{flux field})) that the gauge-invariant fields $(A,E)$ are written in each cell $C_v$ in general as
\be
A=a_v\rd a_v^{-1},\qquad\qquad E=a_v\rd Z_va_v^{-1},
\ee
where we have used a Lie algebra-valued one-form $Z_v\in\Omega^1\big(C_v,\su(2)\big)$ to write $X_v(x)=\rd Z_v(x)$. Since a flat triad must be torsion-free by the above lemmas, we can similarly write a flat triad in general as
\be\label{flat e}
\bar{e}=a_v\rd z_va_v^{-1},
\ee
for some Lie algebra-valued function $z_v\in\Omega^0\big(C_v,\su(2)\big)$. The function $z_v$ provides a set of flat coordinates in $C_v$. Requiring the triad to be invertible places the following condition on $z_v$:
\be\label{z condition}
\epsilon^{abc}\epsilon_{ijk}\partial_az_v^i\partial_bz_v^j\partial_cz_v^k>0.
\ee
The electric field constructed from this triad is given by
\be\label{flat E}
\bar{E}=a_v\left[\rd z_v,\rd z_v\right]a_v^{-1}.
\ee

Consider that we are given a pair $(a_v,Z_v)$ defining an electric field $E$ within a cell $C_v$. Looking at (\ref{flat-gauge}), we seek a gauge field $\phi_v\in \Omega^1\big(C_v,\su(2)\big)$ such that
\be
\rd_A \phi_v=\bar{E}-E=a_v\left([\rd z_v,\rd z_v]-\rd Z_v\right)a_v^{-1}.
\ee
Using $a_v^{-1}\rd_A\phi_va_v=\rd(a_v^{-1}\phi_va_v)$, we can solve for $\phi_v$ to obtain
\be\label{solve phi}
\phi_v=a_v\left([z_v,\rd z_v]-Z_v+\rd g_v\right)a_v^{-1},
\ee
for $g_v\in\Omega^0\big(C_v,\su(2)\big)$. For $x\in\Gamma^*$ the value of $g_v(x)$ is fixed up to an overall constant by the condition $\phi_v(x)=0$:
\be\label{solve g}
g_v(x)=\int_{s(e)}^x\left(Z_v-[z_v,\rd z_v]\right),
\ee
where the integration is along a link $\ell$ in the boundary $\partial F_e$ of the face $F_e$.

We have shown the existence of a gauge field $\phi_v$ taking us from an arbitrary electric field $E\in\mathcal{P}^\mathcal{G}_{\Gamma,\Gamma^*}$ to a flat electric field with vanishing intrinsic and extrinsic curvature in a cell. The next question to ask is whether this choice unique. Since $g_v$ is fixed only on $\Gamma^*$ (and even there only up to a constant), and it is not fixed in $C_v$ or the faces $F_e$, there are many choice of $\phi_v$ which give the transformation $E\rightarrow \bar{E}$. Moreover, any $z_v$ satisfying (\ref{z condition}) gives a flat, invertible triad, so there is not even a unique choice of flat electric field. Therefore the transformation to the flat cell gauge is not unique.

Having found a gauge transformation to a flat electric field in a single cell, we now consider how this transformation affects the geometry at cell boundaries when performing this transformation in all cells of the cellular decomposition. Consider a region $C_{12}\equiv\overline{C_{v_1}}\cup\overline{C_{v_2}}$ formed by two cells and their boundaries. The requirement of continuity of $E$ and $\bar{e}$ at the face $F_e=\overline{C_{v_1}}\cap\overline{C_{v_2}}$ gives conditions at the face:
\be\label{cond z}
\lim_{y \rightarrow x^+}\rd z_2(y)=\lim_{y\rightarrow x^-}h_e^{-1}\rd z_1(y)h_e,\qquad\qquad
\lim_{y \rightarrow x^+}\rd Z_2(x)=\lim_{y\rightarrow x^-}h_e^{-1}\rd Z_1(x)h_e,
\ee
for $x\in F_e$ and $y\in C_{12}$, where $y\rightarrow x^+$ is used to indicate that the coordinate $y$ approaches $x$ from within $C_{v_2}$, and $y\rightarrow x^-$ indicates that $y$ approaches $x$ from within $C_{v_1}$. Using these relations and requiring $\phi_2(x)$ to vanish on the boundary of the face adds another condition:
\be
\lim_{y\rightarrow x^+}\rd g_2(x)=\lim_{y\rightarrow x^+}h_e^{-1}\rd g_1(x)h_e.
\ee
Together, these relations imply that $\phi_2(x)=\phi_1(x)$, so that the gauge field is continuous across faces.

Let us look more closely at the extrinsic curvature on the face $F_e$. In Lemma \ref{K zero} we showed that vanishing torsion implies zero extrinsic curvature. What is the torsion at the face in the flat cell gauge? To answer this question we zoom in to a small neighborhood of $C_{12}$ which contains the face so that we can define a local cartesian coordinate system where $y^0$ is perpendicular to the face, $y^\mu$ for $\mu=1,2$ run parallel to the face, and we set $y^0 = 0$ on the face. In this neighborhood we define a one-form:
\be
\rd z(y)=\rd z_1(y)+\Theta(y^0)\left(h_e^{-1}\rd z_1(y)h_e-\rd z_1(y)\right),
\ee
where $\Theta(y^0)$ is a step function whose value is $0$ for $y\in C_{v_1}$, and $1$ for $y\in C_{v_2}$. The torsion is given by:
\be
\rd_A\bar{e}(y)=\rd_A\left(a(y)\rd z(y)a(y)^{-1}\right)=a(y)\rd\rd z(y)a(y)^{-1}.
\ee
Now, the right hand side of this equation is zero away from the face, but more scrutiny is required at the face. Since $\partial_0\Theta(y^0)=\delta(y^0)$, we see explicitly that $\partial_0\partial_\mu z(y)-\partial_\mu\partial_0z(y) \ne 0$ when $y^0=0$, and therefore $\rd_A\bar{e}\ne0$ at the face. Considering equation (\ref{extrinsic}), this leads to a non-vanishing extrinsic curvature at the faces of the cellular decomposition.

Finally, we close this section with a reconstruction of the flux elements $X_e$ starting from the flat frame field $\bar{e}$. The flux elements in this gauge are given by the simple form:
\be
X^i_e=\f{1}{2}\epsilon^i_{~jk}\int_{F_e}h_{\pi_e}\bar{e}^j\wedge \bar{e}^k h_{\pi_e}^{-1}=\f{1}{2}\epsilon^i_{~jk}\int_{F_e}\rd z_v^j\wedge\rd z_v^k=\f{1}{2}\epsilon^i_{~jk}\int_{\partial F_e}z_v^j \rd z_v^k,
\ee
where we have used the fact that $h_{\pi_e} = a_v^{-1}$.

\subsection{Regge geometries}

\noindent The previous calculation shows that we can think of the phase space $P_\Gamma$ as the phase space of piecewise (metric) flat geometries on $\Sigma\backslash \Gamma^{*}$. Such geometries possess an invertible locally flat metric, with extrinsic curvature concentrated on the faces of the cellular decomposition. This description is reminiscent of Regge geometries. However, it is known that the phase space of loop gravity is bigger than the phase space of Regge geometry \cite{dittrich-ryan}; Regge geometries appear only as a constrained subset. This fact has
triggered the search for the proper geometrical interpretation of the loop gravity phase space, for instance in terms of twisted geometries \cite{freidel-speziale}.

We can now clearly understand the key difference between the phase space of loop gravity and that of Regge geometries. In the flat gauge, the loop gravity phase space corresponds to a cellular decomposition of the spatial manifold $\Sigma$ where the extrinsic curvature is zero within each three-cell $C_v$ but non-zero on the faces $F_e$. The faces do not need to be flat two-surfaces, and may be arbitrarily curved so long as they do not self-intersect and only intersect with other faces along common boundaries. The difference between this setting and a Regge geometry is the arbitrariness in the shape of the faces; the faces are all flat in a Regge geometry.

In order to see how the loop gravity phase space (in the flat gauge) may be reduced to a Regge geometry, we must ask how can the faces be made flat? A necessary condition for a face $F_e$ to be flat is that the boundary $\partial F_e$ is composed of flat links. Since $\Gamma^*$ is the union of all face boundaries, $\Gamma^*$ must consist entirely of flat links in order to obtain a Regge geometry.

Let us go back to the formula for the fluxes that we have derived in the previous subsection:
\be\label{Regge flux}
X^i_e=\f{1}{2}\epsilon^i_{~jk}\int_{\partial F_e}z_v^j\wedge\rd z_v^k,
\ee
where $z_{v}$ is the flat coordinate in the cell $C_v$. One sees that if the links $\ell \in \partial F_e$ are chosen to be flat, then $z_v$ is linear and $\rd z_v$ is constant over $\partial F_e$. This simplifies the expression drastically. Recall that due to the Gauss law, the fluxes are independent of the choice of faces (\ref{same flux}) for fixed $\Gamma^*$. This means that (\ref{Regge flux}) is independent of the choice of face, so that we obtain the same flux whether the face is chosen curved or flat, so long as the boundary of the face is composed of flat links. Indeed, a Regge geometry is given by a unique set of link lengths which can be reconstructed from the fluxes and dihedral angles between links, independently from the choice of faces. Imposing that $\Gamma^*$ be composed entirely of flat links implies that the fluxes can be constructed, using (\ref{Regge flux}), entirely in terms of a discrete piecewise flat geometry \`a la Regge.

In the twisted geometries construction \cite{freidel-speziale} the geometry is seen as flat polyhedra glued together along faces. While two faces that are glued together have the same area, they may generally have different shapes. This means the metric is discontinuous across faces, although it is still possible to define a spin connection \cite{HRWV}. The reduction to a Regge geometry is done using gluing constraints \cite{dittrich-ryan}. These constraints impose that the shapes match by enforcing that corresponding dihedral angles on the face boundaries agree.

In our cellular decomposition there is only one face between neighboring cells, so there is no notion of pairs of faces that must be made to fit together. Once the links of $\Gamma^*$ are made flat, the gluing constraints are automatically satisfied by construction. This means that the set of holonomies and fluxes on a graph can be implemented as a piecewise flat geometry on $\Sigma\backslash\Gamma^*$ by making a particular gauge choice, and corresponds to a Regge geometry if we impose the additional constraint that the edges of $\Gamma^*$ are straight with respect to the flat structure\footnote{This means that $\rd z_v$ is constant on the edges of $\Gamma^*$.}. The phase space of full loop gravity then corresponds to piecewise geometries where this additional restriction is not imposed. In other words, the edges of $\Gamma^*$ do not have to be flat when mapping from the loop gravity phase space to the continuous phase space using the flat cell gauge.

\subsection{Cotangent bundle}

\noindent The result of our construction is that after a choice of gauge, we can express the elements of $P_\Gamma$ as a connection $A$ and an $\su(2)$-valued frame field $e$, which are solutions to
\be
F(A)(x)=0,\qquad\qquad\rd_Ae(x)=0,\qquad\qquad\forall\:x\in\Sigma\backslash\Gamma^*.
\ee
Since $\delta F(A)=\rd_A\delta A$, this is nothing but the cotangent bundle of the space of flat $\SU(2)$ connections on $\Sigma\backslash \Gamma^{*}$.
That is
\be
P_\Gamma= T^*\mathcal{M}_{\Gamma^*},
\ee
where $\mathcal{M}_{\Gamma^*}$ denotes the moduli space of flat connections modulo gauge transformations. This means that at the quantum level we can represent the quantization of holonomies and fluxes in terms of operators acting on holonomies of flat connections. This interpretation has already proposed by Bianchi in \cite{bianchi}. It is interesting to note that this is reminiscent of the geometry considered by Hitchin in \cite{hitchin}.

\subsection{Diffeomorphisms and gauge choices}

\noindent We have seen in subsection \ref{diff-gauge} that diffeomorphisms $\Phi_{\!o}$ connected to the identity that do not move  $\Gamma^*$ and the vertices of $\Gamma$ leave the construction of the holonomy-flux algebra invariant. We have also seen in the beginning of this section that the singular gauge and the flat gauge are diffeomorphism covariant. In general, the construction of $h_e$ and $X_{F_e}$ depends both on $\Gamma$ via the choice of $e$, and on $\Gamma^*$ via the choice of a two-cell $F_{e}$. Now, because of the flatness of the connection, the holonomy does not really depend on the choice of edge $e$, but solely on the choice of the homotopy class of $e$, which itself is left unchanged by diffeomorphisms that are connected to the identity. For the isomorphism between $\mathcal{P}^{\mathcal{G}}_{\Gamma, \Gamma^*}$ and $P^G_\Gamma$, it is interesting to note that the choice of the singular gauge is invariant under a diffeomorphism that does not move $\Gamma$, whereas the choice of the flat gauge is invariant under diffeomorphisms that do not move $\Gamma^*$. Indeed, in the singular gauge the frame field depends on the choice of an edge $e\in\Gamma$, and we have $\Phi^*E=E$ if $\Phi(\Gamma) =\Gamma$. Moreover, under an infinitesimal diffeomorphism $\xi$, the flux becomes
\be
\delta_\xi X_e=\int_{\partial F_e}\iota_\xi\left(h_{\pi_e}(x)E(x)h_{\pi_e}(x)^{-1}\right),
\ee
where $h_{\pi_e}(x)$ is again the holonomy going from the source vertex of the edge $e$ to the point $x$ in $F_e$. We clearly see that this expression vanishes for all $\xi$ when the electric field is in the singular gauge. In the flat gauge, the flux does not depend on $\Gamma$, and the construction is therefore invariant under diffeomorphisms leaving $\Gamma^*$ invariant.

This shows that there is an interesting duality between the two gauges. While the singular gauge respects diffeomorphism invariance with respect to $\Gamma$, the flat one respects diffeomorphism invariance with respect to $\Gamma^*$.

\section{Cylindrical consistency}
\label{sec5}

\noindent An important property of operators in LQG is that of cylindrical consistency associated with a projective family of graphs  \cite{AL}. In a projective family of graphs we have an ordering such that we may write for any two graphs in the family that $\Gamma_1<\Gamma_2$ if $\Gamma_2$ contains all the edges of $\Gamma_1$ in addition to other edges. A cylindrically consistent function is such that the pull-back from $P_{\Gamma_1}$ to $P_{\Gamma_2}$ is identified with the function on the $P_{\Gamma_1}$.

In this section we give a proposal for extending the notion of cylindrical consistency to functionals $\mathcal{O}[A,E]$ of the continuous fields. We analyze to what extent the knowledge of a collection of functions on $P_\Gamma$ for all $\Gamma$ determines a continuous functional. Given a collection of functions $\mathcal{O}_\Gamma\in P_\Gamma$, we now propose an extension of cylindrical consistency to continuous functionals.

\begin{definition}
Suppose that we are given a collection of functions $\mathcal{O}_\Gamma\in P_\Gamma$. We say that such a collection of functions is cylindrically consistent
%
if there exists a continuous functional $\mathcal{O}[A,E]$  such that its restriction on the constraint surface $\mathcal{C}$ is equal to $\mathcal{O}_\Gamma$. That is
\be
\mathcal{O}|_\mathcal{C}[A,E]=\mathcal{O}_\Gamma\big[h_e(A),X_e(A,E)\big].
\ee
\end{definition}

The results presented in the previous sections show that such a continuous functional $\mathcal{O}[A,E]$ is characterized by the the following property:

\begin{proposition}
$\mathcal{O}[A,E]$ is a cylindrical functional if and only if its restriction to the constraint surface $\mathcal{C}$ is invariant under the gauge group $\mathcal{F}_{\Gamma^*}\times\mathcal{G}_\Gamma$ for every pair of dual graphs $(\Gamma,\Gamma^*)$.
\end{proposition}

Indeed, suppose that we have a functional $\mathcal{O}[A,E]$ defined on the phase space $\mathcal{P}$ such that  its restriction to the constraint surface $\mathcal{C}$ is then $\mathcal{O}|_\mathcal{C}[A,E]$, where the field configurations now satisfy $F(A)=0$ outside of the dual graph $\Gamma^*$, and $\rd_AE=0$ outside of the vertices $V_\Gamma$.  $\mathcal{O}[A,E]$ is a cylindrically consistent functional if and only if
\be
\mathcal{O}|_\mathcal{C}\big[g\triangleright A,(\phi,g)\triangleright E\big]=\mathcal{O}|_{{}_\mathcal{C}}[A,E],
\ee
which necessarily implies that $\mathcal{O}[A,E]=\mathcal{O}\big[h_e(A),X_e(A,E)\big]$.

This proposition gives us a powerful criterion to check wether a continuous functional can be represented as a collection of functions associated with $P_{\Gamma}$. For instance, we can analyze the status of geometrical functionals such as area and volume. We know that the continuous expression for the area functional is
\be
\mathbf{A}(S)=\int_S\sqrt{\tilde{E}^i_a\tilde{E}_i^a}.
\ee
One can easily see that even when we restrict this functional to the constraint surface $F(A)=0$ outside $\Gamma^{*}$ and $\rd_{A}E =0$, this functional is \textit{not} invariant under the translations $E\longmapsto E+\rd_{A}\phi$. Therefore, this functional is \textit{not} expressible purely in terms of holonomies and fluxes associated with the graph $\Gamma$. However, in loop quantum gravity, the area operator is expressed as an operator acting on the graph $\Gamma$, and is the quantum version of a function of the fluxes\footnote{For the moment we shall use the covariant fluxes (\ref{new flux}) in this definition, even though the traditional LQG area operator descends from a functional defined using (\ref{usual flux}). We shall comment more on this below.}:
\be
\label{ALQG}
\mathbf{A}_\text{\tiny{LQG}}(S)=\sum_{e|e\cap S\neq0}\sqrt{X^{i}_{e}X_{ei}}.
\ee
Our proposition therefore shows that the LQG area operator does not come from the continuous area functional. This means that we have
\be
\mathbf{A}(S)|_\mathcal{C}-\mathbf{A}_\text{\tiny{LQG}}(S)\neq0.
\ee
So in that sense, the LQG operator is not a proper approximation of the continuous area functional.

This is puzzling since the LQG area operator has been used extensively and derived in many ways. This result thus raises the question of the exact relationship between these two objects. To what extend does the LQG operator capture information about the continuous area functional? Now, since we have the exact relationship between the discrete and continuous phase spaces, we can investigate this question a bit further.

First, let us recall that the continuous and LQG areas are not unrelated. In fact, for any product $h_\Gamma$ of holonomies supported on the graph $\Gamma$, they satisfy
\be
\big\lbrace\mathbf{A}(S)|_\mathcal{C}-\mathbf{A}_\text{\tiny{LQG}}(S),h_\Gamma\big\rbrace=0.
\ee
So even if $\mathbf{A}|_\mathcal{C}-\mathbf{A}_\text{\tiny{LQG}}$ does not vanish, it belongs to the commutant of the holonomy algebra.

The second key remark is that if we have a non-gauge-invariant functional like $\mathbf{A}(S)$, we can promote it to a gauge-invariant functional under $\mathcal{F}_{\Gamma^*}$ by picking up a gauge. This can be done by working with  $\mathbf{A}^\mathcal{T}(S)\equiv\mathbf{A}(S)\big(E(X_e)\big)$ instead of  $\mathbf{A}(S)(E)$, where $\mathcal{T}$ is a gauge choice as described in section \ref{sec4}. Such a functional is by construction invariant under $\mathcal{F}_{\Gamma^*}$, since it depends only on the fluxes. Moreover, the difference between two functionals that differ by a choice of gauge belongs to the commutant of the holonomy algebra:
\be
\big\lbrace\mathbf{A}^\mathcal{T}(S)-\mathbf{A}^{\mathcal{T}'}(S), h_{\Gamma}\big\rbrace=0.
\ee

This implies that the LQG area operator is the quantization of the continuous area functional written in a particular gauge, and as described in section \ref{sec4}, the interpretation of geometry in LQG is given by the singular gauge.
This explains why it can be expressed purely in terms of fluxes.

So far in (\ref{ALQG}) we have considered the covariant flux (\ref{new flux}) rather than the usual definition (\ref{usual flux}). Does this analysis hold for an area operator defined from the traditional definition of flux? In the singular gauge the electric field is given by (\ref{Esing}), and the integral defining the covariant flux (\ref{new flux}) receives a contribution only at the point of intersection between the surface $S$ and the edge $e$. The dependence on $h_\pi$ is traced out in the definition (\ref{ALQG}) so that in the singular gauge, the area functional is the same whether one uses the covariant flux or the usual definition. Therefore, the above analysis is valid for either form of the flux.

Now, what is unclear is to what extent the knowledge of a function in a given gauge allows reconstruction of the continuous functional. Also, if one chooses another gauge, like the flat gauge of spin foam models, we are going to construct a different family of area functions associated with graphs, which will differ from $\mathbf{A}_\text{\tiny{LQG}}$ by an element of the commutant of the holonomy algebra. It is not clear which family of operators (if any) we should use to capture in the most efficient way information about the continuous volume operator.

\section*{Discussion and conclusion}

In this paper, we have shown that the discrete phase space of loop gravity associated with a graph $\Gamma$ can be interpreted as the symplectic reduction of the continuous phase space of gravity with respect to a constraint imposing the flatness of the connection everywhere outside of the dual graph $\Gamma^*$. This allows us to give a clear interpretation of the discrete flux variables as labeling an equivalence class of continuous geometries. The point of view that the discrete data represents a set of continuous geometries has already been advocated in \cite{rovelli-speziale}. Our approach gives a precise understanding of which set or equivalence class of continuous geometries is represented by the discrete geometrical data $(h_e,X_e)$ on a graph. It provides a classical understanding of the work by Bianchi \cite{bianchi}, who showed that the spin network states can be understood as states of a topological field theory living on the complement of the dual graph. It also allows us to reconcile the tension between the loop quantum gravity picture, in which geometry is thought to be singular, and the spin foam picture, in which the geometry is understood as being locally flat. We now see that both interpretations are valid and correspond to different gauge choices in the equivalence class of geometries represented by the fluxes. It gives us a new understanding of the geometrical operators used in loop quantum gravity as gauged fixed operators, and allows us to investigate further the relationship between these operators and the continuous ones. Finally, it opens the way to a classical formulation of loop gravity. We can now face the question of whether the dynamics of classical general relativity can be formulated in terms of these variables. We plan to come back to this issue of defining a loop classical gravity in the future.

\subsection*{Acknowledgments}

\noindent It is a pleasure to thank Valentin Bonzom, Eugenio Bianchi, Johannes Brunnemann, Kirill Krasnov and Karim Noui for discussions and comments.

\end{document}